\documentclass[accepted]{uai2026}
\usepackage{amsmath}
\usepackage{float}
\usepackage{mathtools}
\mathtoolsset{showonlyrefs}

\usepackage{amsthm}
\usepackage{amssymb}
\usepackage{graphicx}
\usepackage{lmodern}
\usepackage{natbib} 
\usepackage{xcolor}
\usepackage{mathrsfs}

\makeatletter
\usepackage{comment}

\usepackage{lastpage} 
\usepackage{extramarks} 
\usepackage[bb=dsserif]{mathalpha} 
\usepackage{epstopdf} 
\usepackage{amsmath}
\usepackage{algorithm}
\usepackage{algcompatible}
\usepackage{bbm}
\usepackage{enumerate}

\usepackage{tikz}
\usepackage{pgfplots}
\usepackage{url}
\usepgfplotslibrary{groupplots,dateplot}
\usetikzlibrary{patterns,shapes.arrows, fadings}
\usetikzlibrary{automata} 
\usetikzlibrary{shapes.geometric}
\usetikzlibrary{calc,shadows}
\usetikzlibrary{positioning} 
\usetikzlibrary{arrows} 
\tikzset{node distance=2.5cm, 
every state/.style={ 
semithick,
fill=gray!10},
initial text={}, 
double distance=2pt, 
every edge/.style={ 
draw,
->,>=stealth', 
auto,
semithick}}
\pgfplotsset{compat=newest}
\newlength\figH
\newlength\figW
\theoremstyle{definition}
\newtheorem{definition}{\protect\definitionname}
\newtheorem{assumption}{\protect\assumptionname}
\newtheorem{lemma}{\protect\lemmaname}

\newtheorem{example}{Example}

\theoremstyle{plain}
\newtheorem{theorem}{\protect\theoremname}
\newtheorem{problem}{Problem}

\newtheorem{corollary}{Corollary}

\DeclareMathOperator{\adj}{Adj}

\DeclareMathOperator{\Beta}{Beta}
\DeclareMathOperator{\diag}{diag}

\newcommand{\pushright}[1]{\ifmeasuring@#1\else\omit\hfill$\displaystyle#1$\fi\ignorespaces}

\makeatother

\providecommand{\assumptionname}{Assumption}
\providecommand{\definitionname}{Definition}
\providecommand{\lemmaname}{Lemma}
\providecommand{\theoremname}{Theorem}
\providecommand{\propositionname}{Proposition}
\newcommand{\norm}[1]{\left\lVert#1\right\rVert}

\newcommand{\E}[1]{\mathbb{E}\left[#1\right]}

\newcommand{\ones}{\mathbb{1}}

\newcommand{\gammaf}[1]{\Gamma\left(#1\right)}

\newcommand{\prob}[1]{\mathbb{P}\left(#1 \right)}

\newcommand{\R}{\mathbb{R}}

\newcommand{\C}{C}

\newcommand{\dir}[1]{\mathcal{M}_{\text{Dir}}^{(k)}(#1)}

\newcommand{\D}{\mathcal{D}}
\newcommand{\N}{\mathbb{N}}
\newcommand{\kld}{\mathcal{K}}

\newboolean{show_comments}
\setboolean{show_comments}{true} 
\ifthenelse{\boolean{show_comments}}{
    \newcommand{\Alex}[1]{\textcolor{blue}{Alex:~#1}} 
    \newcommand{\mh}[1]{\textcolor{red}{MH:~#1}} 
    \newcommand{\brandon}[1]{\textcolor{orange}{bf:~#1}} 
} {
    \newcommand{\Alex}[1]{}
    \newcommand{\mh}[1]{}
    \newcommand{\brandon}[1]{}
}

\newcommand{\new}[1]{#1}

\title{Differentially Private Data-Driven Markov Chain Modeling}
\author[1]{Alexander Benvenuti}
\author[2]{Brandon Fallin}
\author[1]{Calvin Hawkins}
\author[3]{Brendan Bialy}
\author[3]{Miriam Dennis}
\author[2]{Warren Dixon}
\author[1]{Matthew Hale}
\affil[1]{%
    School of Electrical and Computer Engineering\\
    Georgia Institute of Technology\\
    Atlanta, Georgia, USA
}
\affil[2]{%
    Mechanical and Aerospace Engineering Department\\
    University of Florida\\
    Gainesville, Florida, USA
}
\affil[3]{%
    Munitions Directorate\\
    Air Force Research Lab\\
    Eglin Air Force Base, Florida, USA
}

\begin{document}
\maketitle





\begin{abstract}
Markov chains model a wide range of user behaviors. However, generating accurate Markov chain models requires substantial user data, and sharing these models without privacy protections may reveal sensitive information about the underlying user data. We introduce a method for protecting user data used to formulate a Markov chain model. First, we develop a method for privatizing database queries whose outputs are elements of the unit simplex, and we prove that this method is differentially private. We quantify its accuracy by bounding the expected KL divergence between private and non-private queries. We extend this method to privatize stochastic matrices whose rows are each a simplex-valued query of a database, which includes data-driven Markov chain models. To assess their accuracy, we analytically bound the change in the stationary distribution and the change in the convergence rate between a non-private Markov chain model and its private form. Simulations show that under a typical privacy implementation, our method yields less than~$2\%$ error in the stationary distribution, indicating that our approach to private modeling faithfully captures the behavior of the systems we study. 
\end{abstract}

\section{INTRODUCTION}\label{sec:intro}
Markov chain modeling is a powerful tool for understanding a variety of behaviors including domestic time use~\citep{widen2009combined}, traffic patterns~\citep{gong2011iterative}, and internet browsing~\citep{dong2022personalized}. 
Markov chains are often to model systems with random transitions among a finite set of states. In data-driven settings, transition probabilities are typically derived from databases of observed user behavior, as in~\cite{rendle2010factorizing} and~\cite{lundstrom2016detecting}.
The development of these models often counts transitions between states in user data, and sharing the resulting models may leak information about the data used to generate the transition probabilities. 
Patterns such as home occupancy~\citep{lundstrom2016detecting} or individual shopping behavior~\citep{rendle2010factorizing} may be inferred even from aggregated data. Thus, without some form of protection, these models present significant privacy risks for users.
However, these models fundamentally require user data, which motivates our interest in protecting privacy while preserving model accuracy. 

In this paper, we develop a framework for producing privacy-preserving Markov chain models that protects databases used to compute transition probabilities. We first develop a method for privatizing database queries whose outputs are stochastic vectors, i.e., vectors with non-negative entries that sum to~$1$. We then extend this method to stochastic matrices in which every row is equal to a query whose output is a stochastic vector.

We develop this framework using differential privacy. Differential privacy is a statistical notion of privacy designed to preserve the privacy of sensitive databases when queried~\citep{dwork2006calibrating}. We use differential privacy because (i) compositions of differentially private data remain differentially private, and (ii) it is immune to post-processing, in the sense that post-hoc computations on private data do not weaken differential privacy. These properties have led to the wide use of differential privacy for federated learning~\citep{geyer2017differentially, agarwal2021skellam, chen2022fundamental, noble2022differentially}, optimization~\citep{hsu2014differential,wang2016differentially,munoz2021private,benvenuti2024guaranteed}, planning~\citep{gohari2020privacy, chen2023differential,benvenuti2023differentially}, and machine learning~\citep{ponomareva2023dp, jayaraman2019evaluating, blanco2022critical}.

In the context of Markov chains, property (i) allows us to privatize each row in a stochastic matrix to form a private Markov chain model, and property (ii) allows the private model to be shared and analyzed without harming the privacy of the underlying database. Additionally, the differential privacy literature has shown that highly accurate query outputs can sometimes be produced while enforcing strong privacy~\citep{dwork2014algorithmic}. Thus, we use differential privacy to generate and share accurate Markov chain models while still protecting users with strong privacy.

To privatize stochastic vectors, we extend the Dirichlet mechanism first introduced in~\cite{gohari2021differential}. Common privacy mechanisms such as the Gaussian and Laplace mechanisms~\citep{dwork2014algorithmic, mcsherry2009privacy} are ill-suited for stochastic vectors because they add noise with infinite support, which may cause private output vectors to no longer be stochastic. The resulting vectors could be projected onto the unit simplex to enforce that their entries are non-negative and sum to~$1$, but~\cite{gohari2021differential} showed such projections result in poor accuracy. We fundamentally differ from~\cite{gohari2021differential} because
we privatize databases that are used to compute simplex-valued queries,
whereas~\cite{gohari2021differential} privatized elements of the simplex
themselves (without there being an underlying database that generated them). \new{Our approach provides event-level privacy to data, in the sense that the presence of individual events is protected by differential privacy. 
This form of privacy is intended to be used for datasets in which individual users are not tracked over time, such as the taxi cab dataset in~\cite{nyc2025traffic}, which
we use in simulation in Section~\ref{sec:sims}. 
}


Of course, we expect the asymptotic and transient behaviors of a privatized Markov chain to differ from those of the comparable non-private Markov chain. To assess the impact of privacy on asymptotic behavior, we bound the change in the stationary distribution (i.e., the long-run probability
distribution of the Markov chain's state) between a private and non-private Markov chain. To assess the impact on transient behavior, we bound the change in ergodicity coefficient (which is a proxy for the convergence rate to the stationary distribution) between a private and non-private Markov chain. These bounds quantify privacy's impact on Markov chain models, and they can be used to tune privacy parameters based on accuracy. 

To summarize, our contributions are as follows:
\begin{itemize}
    \item We develop a framework for privatizing queries whose outputs are stochastic vectors (Theorem~\ref{thm:dp}).
    \item We bound the Kullback-Leibler divergence between a non-private stochastic vector and its privatized form (Theorem~\ref{thm:kld} and Corollary~\ref{cor:kld}).
    \item Using~$1$, we develop a framework for privatizing the computation of Markov chain transition probabilities computed from user data (Theorem~\ref{thm:mc_dp}).
    \item We bound the changes in asymptotic and transient behaviors between the original, non-private Markov chain and its privatized form by bounding the differences in the stationary distribution and the ergodicity coefficient (Theorems~\ref{thm:stat} and~\ref{thm:ergodic}).
    \item We empirically validate the impact of privacy on performance on two real datasets: a class grade distribution and transit data from New York City taxis. Results show less than~$2\%$ error in the stationary distribution with~$(3.73, 6\times 10^{-6})$-differential privacy (Section~\ref{sec:sims}).
\end{itemize}
\textbf{Related Work}
In this work, we build on the original Dirichlet mechanism from~\cite{gohari2021differential}. More recently,~\cite{ponnoprat2023dirichlet} and \cite{ chitra2022differential} have developed alternate formulations of the Dirichlet Mechanism. However, these formulations introduce a bias in private outputs, which is undesirable for the Markov chain models we consider. Thus, we base our approach on~\cite{gohari2021differential} to yield unbiased private outputs.

Privacy for Markov processes has been extensively studied~\citep{gohari2020privacy,chen2023differentialsymbolic, chen2023differential,fallin2023differential, benvenuti2023differentially}, though only~\cite{gohari2020privacy} and~\cite{fallin2023differential} privatize the transition probabilities, which we do in this work. However, those works consider the probabilities themselves as the sensitive data and apply input perturbation privacy to them. In contrast, we consider transition probabilities that are computed from data, and we privatize the underlying data, not the values of the probabilities. Thus,~\cite{gohari2020privacy} and \cite{ fallin2023differential} represent special cases of this paper
where
the transition probabilities are the sensitive database.  

\textbf{Notation} 
For~$N\in\mathbb{N}$, we use~$[N] := \{1,\ldots, N\}$. 
We use~$|\cdot|$ to denote the cardinality of a set. 
We use $\diag(a_1, \ldots, a_n)$ to denote an~$n\times n$ diagonal matrix with diagonal entries~$a_1, \ldots, a_n$. We use~$\mathbb{I}\{\cdot\}$ to denote the indicator function and~$\ones_n$ to be a vector of all~$1$'s of length~$n$. 
For a random variable $X$, we use $\E{X}$ to denote its expectation. We use~$\kld(q||\ell) = \sum_{i=1}^n q_i\log\big(\frac{q_i}{\ell_i}\big)$ to denote the KL divergence between~$q$ and~$\ell$,
where we use the convention~$0\log\big(\frac{0}{0}\big) = 0$. We use the following special functions:
The Gamma function~$\gammaf{x} = \int_0^{\infty}t^{x-1}e^{-t}dt$, the Beta function~$\Beta(x_1, x_2) = \int_0^1 t^{x_1-1}(1-t)^{x_2-1}dt = \frac{\gammaf{x_1}\gammaf{x_2}}{\gammaf{x_1+x_2}}$, the multi-variate Beta function
    $B(y) = \prod_{i=1}^n \gammaf{y_i}(\gammaf{\sum_{i=1}^n y_i})^{-1}$, 
 and the digamma function~$\psi(x) = \frac{\Gamma'(x)}{\gammaf{x}}$.

\section{PRELIMINARIES AND PROBLEM FORMULATION}
In this section, we provide background on Markov chains and differential privacy, define the structure of the databases and queries we consider, and formalize problem statements.
%
\subsection{Unit Simplex and Stochastic Matrices}
\begin{definition}[Unit Simplex]
    Let $n\in \mathbb{N}$. The unit simplex in~$\mathbb{R}^{n}$ denoted by~$\Delta_{n}$ is the set of all element-wise non-negative vectors of length~$n$ whose entries sum to~$1$, where
    \begin{equation}
        \Delta_{n}=\Big\{ x\in\mathbb{R}^{n}\ \mid\ \sum_{i=1}^{n}x_{i}=1,\ x_{j}\geq0\ 
        \text{for all}\ j\in[n]\Big\}. \!\!\!\!\qed 
    \end{equation}
\end{definition}
    We also refer to elements of the unit simplex as ``stochastic vectors". We also consider the bordered unit simplex.
%
\begin{definition}[Bordered Unit Simplex]\label{dfn:bordered-simplex}
    Fix~$n\geq 3$ and~$\eta>0$. Then the bordered unit simplex is 
    \begin{equation}
        \Delta_{n}^{(\eta)}=\Big\{ x\in\Delta_{n}\,\mid
        x_{i}\geq\eta\text{ for all }i\in [n] \Big\}. 
        \qed
    \end{equation}     
\end{definition}
The bordered unit simplex is a subset of~$\Delta_n$ whose components are a sufficient distance from~$0$. 


A matrix whose rows are all stochastic vectors is called a ``stochastic matrix". For~$n\in\mathbb{N}$, we define~$\mathbb{S}_n$ 
    as the set of all~$n\times n$ stochastic matrices.
\subsection{Markov Chains}
Consider a discrete time stochastic process~$(Y_t)_{t\in \N}$ on the state space~$\mathcal{Y}$. The process~$(Y_t)_{t\in \N}$ is said to be a discrete time Markov chain if the sequence satisfies the Markov property, namely
    $\prob{Y_{t+1} \!= \!y_{t+1}\!\! \mid\! Y_1\! = \!y_1, \!\cdots\!,\! Y_t \!=\! y_t} = \prob{Y_{t+1}\!\!\mid\! Y_t \!=\! y_t}\!,$
which states that the conditional probability of transitioning from state~$y_t$ to~$y_{t+1}$ is independent of the sequence of states prior to~$t$. See~\citep{norris1998markov} for a detailed account of Markov chains. 
We study Markov chains that are (i) finite, (ii) irreducible, and (iii) homogeneous. These are Markov chains in which (i) the state space~$\mathcal{Y}$ is finite, (ii) any state~$y\in\mathcal{Y}$ can be reached from any other state~$z\in\mathcal{Y}$ with positive probability, and (iii) the transition probabilities are independent of time. Given a finite-state Markov chain, we construct the transition probability matrix~$G$ such that~$G_{ij} = \prob{Y_{t+1} = y_j\mid Y_{t} = y_i}$. Row~$i$ of the matrix~$G$, which we denote~$G_i$, is the probability distribution of transitioning from state~$y_i$ to any other state in~$\mathcal{Y}$. If~$|\mathcal{Y}| = n$, then~$G\in\mathbb{S}_n$.

\subsection{Differential Privacy}
We use differential privacy to provide privacy to a database~$\D$, which we partition as~$\D= \{D_1,\ldots, D_m\}$, where each set~$D_i$
contains~$N_i\in\mathbb{N}$ entries and~$D_i = \{d^i_1, \ldots, d^i_{N_i}\}$. The goal of differential privacy is to make ``similar" data appear ``approximately indistinguishable" by using a randomized mapping or ``privacy mechanism". Given two databases~$\D$ and~$\D'$, the notion of ``similar" is defined by an adjacency relation~\citep{dwork2014algorithmic}.

\begin{definition}[Adjacency]\label{def:adj}
Consider two databases~$\D = \{D_1, \ldots, D_m\}$ and $\D' = \{D'_1, \ldots, D'_m\}$, where~$D_i$ and~$D_i'$ each
contain~$N_i\in\mathbb{N}$ entries 
for all~$i \in [m]$. Suppose~$D_i = \{d^i_1, \ldots, d^i_{N_i}\}$
and similar for~$D_i'$. 
We write~$\adj(\D, \D')=1$ if 
(a) there exist indices~$i,j$ such that~$d^i_j\neq d^{'i}_j$ and 
(b) for all~$ (\ell, k)\neq (i,j)$, we have~$d^{\ell}_k = d^{'\ell}_k$. 
That is,~$\mathcal{D}$ and~$\mathcal{D}'$ are adjacent if they
differ in one entry. 
We write~$\adj(\D, \D') = 0$ otherwise.\hfill$\blacktriangle $
\end{definition}
A privacy mechanism~$\mathcal{M}(\cdot)$ is a randomized mapping that enforces differential privacy. 

\begin{definition}[Differential Privacy; \citep{dwork2014algorithmic}]\label{def:dp}
    Fix a probability space~$(\Omega, \mathcal{F}, \mathbb{P})$. Let~$\delta\in[0, \frac{1}{2})$ and $\epsilon>0$ be given. A mechanism~$\mathcal{M}:\R^{n}\times\Omega\to\R^m$ is~$(\epsilon, \delta)$-differentially private if for all~$\D, \D'$ that are adjacent in the sense of Definition~\ref{def:adj},  we have~$\prob{\mathcal{M}(\D)\in T} \leq e^\epsilon \prob{\mathcal{M}(\D')\in T}+\delta$
             for all Borel measurable sets~$T\subseteq \R^m$. \hfill$\blacktriangle $
        
\end{definition}

The parameters~$\epsilon$ and~$\delta$ quantify the strength of privacy, and smaller~$\epsilon$ and~$\delta$ imply stronger privacy.


\subsection{Databases and Queries}
We consider databases~$\D$ that are partitioned via~$\D= \{D_1,\ldots, D_m\}\in\mathfrak{D}$, where~$\mathfrak{D}$ is the set of all possible database realizations. 
We sometimes consider the degenerate case where~$\D = \{D\}$. To differentiate between these cases, we use~$D$ without a subscript 
to mean~$\D = D = \{d_1,\ldots, d_N\}$.
A database entry~$d^i_j$ may come \new{from} one of~$n$ categories in a category set~$\rho = \{\rho_1, \ldots, \rho_n\}\in\mathcal{P}$, where~$\mathcal{P}$ is the collection of all category sets. We use~$d^i_j\models \rho_k$ to indicate that entry~$d^i_j$ is from category~$\rho_k$.
\begin{assumption}\label{ass:sat}
     Each database entry~$d^i_j$ belongs to exactly one category~$\rho_k\in\rho$.
\end{assumption}
Assumption~\ref{ass:sat} is mild since the categories may represent physical events that cannot happen concurrently. 

\begin{example}[City Traffic]\label{ex:traffic}
    Consider a database~$D$ that contains one record~$d_j$ for each time a driver turns from a street called First Street onto another street.
    Then~$\mathfrak{D}$ is the set of all traffic patterns in the city. We use~$\mathcal{P}$ to represent all streets in the city, and the category set~$\rho$ is the collection of streets that intersect First Street 
    that drivers may turn onto. In this scenario, Assumption~\ref{ass:sat} says that a single driver must turn off somewhere and cannot turn off onto  multiple different roads. \hfill $\triangle$
\end{example}

For a database~$\mathcal{D} = \{D\}$ with~$N$ entries, 
consider the counting query~$\C: \mathfrak{D}\times \mathcal{P} \to \mathbb{R}^n$ with category set~$\rho = \{\rho_1, \ldots, \rho_n\}$, defined as
\begin{equation}\label{eq:count}
    \!\!\!\C(D, \rho) \!=\!\frac{1}{N}\!\begin{bmatrix}
        \sum_{i=1}^N\! \mathbb{I}\{d_i\!\models\!\! \rho_1\},\!\cdots\!,\!\sum_{i=1}^N \!\mathbb{I}\{d_i\!\models\!\! \rho_n\}\!
    \end{bmatrix}^{\!T}\!\!\!\!.
\end{equation}
In words,~$C(D, \rho)$ is a vector whose entries are the fraction of records in~$D$ that belong to category~$\rho_i$ for each~$i\in[n]$.

Next, we consider the more general case where~$\mathcal{D} = \{D_1,\ldots, D_m\}$, where each~$D_i$ contains~$N_i$ entries and we have the category set~$\rho = \{\rho_1, \ldots, \rho_n\}$. We then consider the mapping~$P: \mathfrak{D} \times \mathcal{P} \to \mathbb{R}^{m\times n}$, where each row of the output matrix is a count vector for a database~$D_i$, namely 
\begin{equation}\label{eq:transitions_from_D}
    P(\mathcal{D}, \rho) = \begin{bmatrix}
        \C(D_1, \rho) & \cdots& \C(D_m, \rho)
    \end{bmatrix}^T.
\end{equation}
When~$m = n$, the matrix~$P(\mathcal{D}, \rho)$ in~\eqref{eq:transitions_from_D} is a stochastic matrix, and it represents the transition probability matrix of a Markov chain. One interpretation of~$P(\mathcal{D}, \rho)$ is as follows:~$D_i$ encodes all of the transitions of users from state~$i$ to some other state, and~$\rho = \{\rho_1, \ldots, \rho_n\}$ is the set of states that they can transition to. Then~$\sum_{j=1}^{N_{i}}\ones\{d_j^i\models \rho_k\}$ is the count of how many users transitioned from state~$i$ to state~$k$. Accordingly, each entry in~$P(\mathcal{D}, \rho)$ is the empirical transition probability in a Markov chain.

\subsection{Problem Statements}
In this work, we solve the following problems:
\begin{problem}\label{prob:count}
    Develop a mechanism to privatize~$\D$ when computing count vectors~$\C(\D, \rho)$ in the form of~\eqref{eq:count}.
\end{problem}
\begin{problem}\label{prob:sim_vec}
    Bound the similarity between a privately generated count vector~$\tilde{\C}$ and the original, non-private count vector~$\C(\D, \rho)$.
\end{problem}
\begin{problem}\label{prob:MC}
    Develop a privacy mechanism to privatize~$\D$ when computing the transition matrix of a Markov chain,~$P(\D, \rho)$.
\end{problem}
\begin{problem}\label{prob:MC_utility}
     Bound the difference in the asymptotic behavior (i.e., stationary distribution) and the transient behavior (i.e., convergence rate) between a Markov chain and its private form.
\end{problem}


\section{Dirichlet Mechanism for Differential Privacy of Simplex-Valued Queries}\label{sec:dm-dp}
In this section, we solve Problems~\ref{prob:count} and~\ref{prob:sim_vec}. We consider the case where~$\D = \{D\}$ and denote this as~$D$. To solve Problem~\ref{prob:count}, we extend the Dirichlet mechanism, first introduced in~\citet{gohari2021differential}.
\begin{definition}[Dirichlet Mechanism;~\citep{gohari2021differential}]\label{def:dmech}
    A Dirichlet mechanism with parameter~$k>0$, which we denote~$\mathcal{M}_{\text{Dir}}^{(k)}$, takes a vector~$p\in\Delta^{\circ}_n$ as input and outputs~$x\in\Delta_n$ according to the Dirichlet distribution centered on~$kp$. That is,
        \begin{equation}
        \prob{\mathcal{M}_{\text{Dir}}^{(k)}(p) \!=\! x}\!\! = \!\!\frac{1}{B(kp)}\!\prod_{i=1}^{n-1}x_i^{kp_i-1}\Bigg(\!1\!-\!\sum_{i=1}^{n-1}x_i\!\Bigg)^{\!\!kp_n-1}\!\!.\!\!\! \qed
    \end{equation}
\end{definition}
The Dirichlet mechanism takes a sensitive vector~$p$ and outputs a private vector~$\tilde{p}$. The level of privacy is tuned using the parameter~$k$. 
 The Dirichlet mechanism \new{in Definition~\ref{def:dmech}} satisfies an alternate notion of privacy known as probabilistic differential privacy. To distinguish the two notions, we say ``conventional differential privacy" to refer to Definition~\ref{def:dp}.


\begin{definition}[Probabilistic Differential Privacy;~\citep{machanavajjhala2008privacy}, \citep{meiser2018approximate}]\label{def:pdp}
    Let~$\mathcal{M}$ be a randomized privacy mechanism, 
    and let~$\mathcal{S}$ be the set of possible outputs of the mechanism. 
    The mechanism~$\mathcal{M}$ satisfies~$(\epsilon,\delta)$-probabilistic differential privacy if the set~$\mathcal{S}$ can be partitioned into two disjoint sets~$\Omega _1$ and~$\Omega _2$
    such that for all~$\D$, we have~$\prob{\mathcal{M}(\D)\in\Omega _2}\leq\delta,$
    %
    and for all~$\D'$ adjacent to~$\D$ and all~$S\in\Omega _1$, we have~$\log\Big( \frac{\prob{ \mathcal{M}(\D)=S }}
        {\prob{\mathcal{M}(\D')=S}} \Big)\leq\epsilon.$ \hfill$\blacktriangle$
\end{definition}
In this work, we partition the output space~$\Delta_n$ into the sets 
    $\Omega_1 = \{x\in\Delta_n\mid x_i\geq \gamma \text{ for all }i\in [n]\}$
and
    $\Omega_2 = \{x\in\Delta_n \mid x\not\in \Omega_1\}$. 

\begin{assumption}\label{ass:gam}
   When generating private outputs in~$\Delta_n$,~$\gamma$ satisfies~$\gamma\leq\frac{1}{n-1}$.
\end{assumption}
If a mechanism provides $(\epsilon, \delta)$-probabilistic differential privacy, then it provides conventional~$(\epsilon, \delta)$-differential privacy~\citep{machanavajjhala2008privacy}. As mentioned in~\cite{meiser2018approximate}, post-processing private data generated by mechanisms that satisfy~$(\epsilon, \delta)$-probabilistic differential privacy only preserves conventional~$(\epsilon, \delta)$-differential privacy. As a result, we use Definition~\ref{def:pdp} only as a means to design a mechanism that satisfies Definition~\ref{def:dp}. 

To solve Problem~\ref{prob:count}, we input~$\C(D, \rho)$ into the Dirichlet mechanism to compute a privatized form~$\tilde{C}$.
We must then compute the privacy parameters~$\epsilon$ and~$\delta$ from Definition~\ref{def:pdp}. The mechanism in \new{Definition~\ref{def:dmech}} does not take a database as input, and we must extend that mechanism to the case where the queries are functions of a database. This extension primarily involves~$\epsilon$.
See the Supplementary Material, Section~\ref{ap:delta} for details on computing~$\delta$.


\subsection{Computing~\texorpdfstring{$\epsilon$}{epsilon}}
To compute~$\epsilon$, we first have two assumptions that ensure the boundedness of the distribution of~$\tilde{C}$.
 \begin{assumption}\label{ass:simp}
    In the bordered unit simplex~$\Delta_{n}^{(\eta)}$ from Definition~\ref{dfn:bordered-simplex}, we require~$\eta \in (0, \frac{1}{4})$.
\end{assumption}
\begin{assumption}\label{ass:k}
    The Dirichlet parameter~$k$ from Definition~\ref{def:dmech} satisfies
        $k\geq \frac{3}{2\eta}.$
\end{assumption}
    Assumption~\ref{ass:simp} is easily enforced by the user selecting an appropriate~$\eta$, since~$\min_{i\in[n]} p_i \geq \eta$ by the construction of~$\C(D, \rho)$.
    Assumption~\ref{ass:k} is also enforced by the user choosing a suitable value for~$k$. Assumption~\ref{ass:k} highlights that not every count vector can be made arbitrarily private. Instead, there exists some strongest level of privacy (i.e., a smallest~$\epsilon$) for a given~$\C(D, \rho)$. We see in Section~\ref{sec:sims} that this strongest level of privacy is often well within the range of typically desired~$\epsilon$ values.
Below we use $\mathcal{C}(\rho)$ to denote the set of all possible outputs of~$\C(D, \rho)$ over all databases of size~$N$ whose entries satisfy Assumption~\ref{ass:sat}. Next we provide~$(\epsilon, \delta)$-differential privacy to~$\C(D, \rho)$. 
All proofs can be found in the Supplementary Material, Section~\ref{ap:proofs}.

\begin{theorem}[Solution to Problem~\ref{prob:count}]\label{thm:dp}
    Fix a database~$D$ with~$N\in\mathbb{N}$ entries, a category set~$\rho$ such that~$|\rho| = n\in\mathbb{N}$, and~$\eta>0$.
    Let~$\C(D, \rho)\in \Delta_{n}^{(\eta)}$ be the count vector as defined in~\eqref{eq:count}. Let Assumptions~\ref{ass:sat}-\ref{ass:k} hold. Then the Dirichlet mechanism with parameter~$k>0$ and input~$\C(D, \rho)$ is~$(\epsilon, \delta)$-probabilistic differentially private, where
        \begin{multline}\label{eq:epsilon}
        \epsilon \geq \log\left(\frac{\Beta(k\eta, k(1-2\eta))}{\Beta\left(k\left(\eta+\frac{1}{N}\right), k(1-2\eta-\frac{1}{N})\right)}\right) +\\\frac{k}{N}\log\left(\frac{1-(n-1)\gamma}{\gamma}\right)
    \end{multline}
    
    and
    \begin{equation}\label{eq:delta}
        \delta = 1-\min_{\C(D, \rho)\in\mathcal{C}(\rho)}\prob{\mathcal{M}_{\text{Dir}}^{(k)}(\C(D, \rho))\in\Omega_1}.
    \end{equation}
\end{theorem}

Theorem~\ref{thm:dp} shows for the first time that the Dirichlet mechanism in Definition~\ref{def:dmech} provides~$(\epsilon, \delta)$-probabilistic (and thus conventional) differential privacy for simplex-valued queries of databases. \new{
A typical differential privacy implementation begins by choosing a value of~$\epsilon$
to implement. Theorem~\ref{thm:dp} 
and Assumption~\ref{ass:k} together show that the smallest possible~$\epsilon$
here is 
$\epsilon^* = \log\Big(\frac{\Beta(\frac{3}{2}, \frac{3}{2\eta}-3)}{\Beta\left(\frac{3}{2}+\frac{3}{2\eta N}, \frac{3}{2\eta}-3-\frac{3}{2\eta N}\right)}\Big) + \frac{3}{2\eta N}\log\Big(\frac{1-(n-1)\gamma}{\gamma}\Big)$.
A user can select any~$\epsilon \geq \epsilon^*$, and that choice can be implemented
by tuning~$k$. This workflow therefore largely mirrors the standard workflow
for implementing differential privacy. 
}
Algorithm~\ref{algo:private_count} provides an overview of our method for privatizing~$D$ when computing~$\C(D, \rho)$. 

\begin{algorithm}[t]
    \caption{Privatizing counts of sensitive data}
    \label{algo:private_count}
    \begin{algorithmic}[1]
    \STATEx\textbf{Inputs}: ~$D$, $\rho$, $N$, $\gamma$, $\eta$, $k$ 
    \STATEx\textbf{Outputs}: Privacy-preserving count $\tilde{\C}$, privacy parameters $\epsilon$, $\delta$  
    \vspace{1mm}
    \STATE Compute the count~$\C(D, \rho)$ using~\eqref{eq:count}
    \STATE Generate~$\tilde{\C} = \dir{\C(D, \rho)}$
        \STATE Compute~$\epsilon$ and~$\delta$ via~\eqref{eq:epsilon} and~\eqref{eq:delta}
    \end{algorithmic}
\end{algorithm}

To analyze the accuracy of Algorithm~\ref{algo:private_count}, we seek to compare ~$\C(D, \rho)$ and~$\tilde{\C}$. Both~$\C(D, \rho)$ and~$\tilde{\C}$ can be interpreted as probability distributions because they are both in~$\Delta_n$.  One way to measure the similarity of two probability distributions is the KL divergence, and
we solve Problem~\ref{prob:sim_vec} by computing the KL divergence between~$\C(D, \rho)$ and~$\tilde{\C}$.

\begin{theorem}[Solution to Problem~\ref{prob:sim_vec}]\label{thm:kld}
    Let the conditions of Theorem~\ref{thm:dp} hold and let~$C = C(D, \rho)$. Then
    \begin{equation}\label{eq:kld}
        \E{\kld(C || \tilde{\C})} = \sum_{i=1}^n C_i(\log(C_i)+\psi(k)-\psi(kC_i)).
    \end{equation}
\end{theorem}
The result in~\eqref{eq:kld} depends on the sensitive count vector~$\C(D, \rho)$, which means that it cannot be computed by someone who only has access to the private output~$\tilde{C}$. To develop a bound that holds independent of the realization of the sensitive data, we upper bound~\eqref{eq:kld} using the KL divergence-maximizing case, which is~$\C(D, \rho) = [\frac{1}{N}{\ones^T_{n-1}},~1-\frac{n-1}{N}]^T$.
\begin{corollary}[Alternative Solution to Problem~\ref{prob:sim_vec}]\label{cor:kld}
    Let the conditions of Theorem~\ref{thm:kld} hold. Let~$\zeta(x) = \log\big(\frac{x+1}{N}\big)-\psi\big(\frac{(x+1)k}{N}\big)$.
    Then~$\E{\kld(\C(D, \rho)||\tilde{\C})}\leq \frac{n-1}{N}\zeta(0) + \frac{N-n+1}{N}\zeta(N-n)+\psi(k).$
        
\end{corollary}

Corollary~\ref{cor:kld} provides users with a tool to tune the average error between a private count vector~$\tilde{\C}$ and the original, non-private count vector~$\C(D, \rho)$. Section~\ref{sims:class} shows the tightness of the Corollary~\ref{cor:kld} relative to Theorem~\ref{thm:kld}.
Supplementary Material Section~\ref{app:brandon_results} provides additional accuracy results for the privatized count vector~$\tilde{\C}$. 
Next we apply Algorithm~\ref{algo:private_count} to a database to privatize stochastic matrices. 

\section{DIFFERENTIAL PRIVACY FOR MARKOV CHAIN MODELING}\label{sec:MC}
In this section we solve Problem~\ref{prob:MC}. Each row of a Markov chain's transition matrix can be computed with a counting query on the unit simplex. We therefore use the mechanism developed in Section~\ref{sec:dm-dp} to privatize each row of a transition matrix individually. 
We apply a property known as ``parallel composition"~\citep{mcsherry2009privacy} to produce a privacy-preserving Markov chain model that protects the database~$\mathcal{D}$ with~$(\epsilon, \delta)$-differential privacy.

Parallel composition of~$\epsilon$-differentially private mechanisms was first proven in~\cite{mcsherry2009privacy}. However, it has since been used in the literature to develop~$(\epsilon, \delta)$-differential privacy mechanisms~\citep{ponomareva2023dp}, often without proof. Incomplete claims about parallel composition for~$(\epsilon, \delta)$-differential privacy have also appeared in the literature~\citep[Theorem 3]{chitra2022differential}. 

We therefore state that parallel composition holds for~$(\epsilon, \delta)$-probabilistic differentially private mechanisms, and we present what is, to the best of our knowledge, the first proof of this result in the Supplementary Material, Section~\ref{ap:proofs}.

\begin{lemma}[Parallel Composition for Probabilistic Differential Privacy]
    \label{lem:parallel-comp-conventional}
    Consider a database~$\D$ partitioned into disjoint subsets~$D_1, \ldots, D_m$,
    and suppose that there are privacy mechanisms~$\mathcal{M}_1, \ldots, \mathcal{M}_m$,
    where~$\mathcal{M}_i$ satisfies~$(\epsilon_i, \delta_i)$-probabilistic differential privacy. Then the release of~$\mathcal{M}(\D) = (\mathcal{M}_1(D_1), \ldots, \mathcal{M}_m(D_m))$
    provides~$\D$ with conventional~$(\max_{i\in[m]} \epsilon,\max_{i\in[m]}\delta)$-differential privacy. \hfill$\blacklozenge$     
\end{lemma}



Using Lemma~\ref{lem:parallel-comp-conventional}, we preserve the privacy of~$\mathcal{D}$ by privatizing the rows of~$P(\mathcal{D}, \rho)$ individually to generate a private transition matrix~$\tilde{P}$. That is, we compute
\begin{equation}\label{eq:ptilde}
    \!\!\!\tilde{P}
    =\begin{bmatrix}
        \mathcal{M}_{\text{Dir}}^{(k_1)}(\C(D_1, \rho)),~  \cdots,~ \mathcal{M}_{\text{Dir}}^{(k_n)}(\C(D_n, \rho))
    \end{bmatrix}^T\!\!.
\end{equation}

We now state the privacy guarantee of computing~$\tilde{P}$.

\begin{theorem}[Solution to Problem~\ref{prob:MC}]\label{thm:mc_dp}
    Fix a database~$\mathcal{D} = \{D_1, D_2,\ldots, D_n\}$, a category set~$\rho$ such that~$|\rho| = n\in\mathbb{N}$, and  Dirichlet parameters~$k_i>0$ for~$i\in[n]$.  Let~$P(\mathcal{D}, \rho) \in \mathbb{S}_n$ be defined as in~\eqref{eq:transitions_from_D}. Let Assumptions~\ref{ass:sat}-\ref{ass:k} hold. Then computing~$\tilde{P}\in\mathbb{S}_n$ using~\eqref{eq:ptilde} keeps~$\mathcal{D}$~$(\max_{i\in[n]}\epsilon_i, \max_{i\in[n]}\delta_i)$-differentially private, where~$\epsilon_i$ is from~\eqref{eq:epsilon} and~$\delta_i$ is from~\eqref{eq:delta}. 
\end{theorem}

\begin{algorithm}[t]
    \caption{Private Markov chain modeling}
    \label{algo:private_MC}
    \begin{algorithmic}[1]
    \STATEx \textbf{Inputs}: ~$\D$, $\rho$, $N_i$, $\gamma_i$, $\eta_i$, $k_i$ for $i\in[n]$
    \STATEx \textbf{Outputs}: Privacy-preserving Markov chain model $\tilde{P}$, privacy parameters $\epsilon$, $\delta$
    \vspace{1mm}
    \FORALL{$i\in[n]$}
         \STATE Compute $\tilde{\C}^i, \epsilon_i, \delta_i$ with Algorithm 1
    \ENDFOR
    \STATE Generate~$\tilde{P}$ by setting row~$i$ equal to~$\tilde{\C}^i$ for~$i\in[n]$
    \STATE Compute~$\epsilon = \max_{i\in[n]} \epsilon_i$ and~$\delta = \max_{i\in[n]} \delta_i$
    \end{algorithmic}
\end{algorithm}
Algorithm~\ref{algo:private_MC} presents our unified framework for computing a privacy-preserving Markov chain model. 

\section{STATIONARY DISTRIBUTION AND CONVERGENCE RATE}\label{sec:utility}
In this section we solve Problem~\ref{prob:MC_utility}.
At each time~$t$, there exists a probability distribution over the state space of the Markov chain, which we denote~$\mu_t\in\Delta_n$. 
Let~$\mu_0\in\Delta_n$ denote the initial distribution. The distribution~$\mu_t$ can be computed by right-multiplying by the transition matrix~$G$ to find~$\mu_t^T = \mu_{t-1}^T G$.

We analyze the steady-state behavior of a Markov chain by analyzing the limiting behavior of~$\mu_t$. When a Markov chain is finite, irreducible, and homogeneous,~$\lim_{t\to\infty}\mu_t$ exists and is called the ``stationary distribution", which we denote by~$\pi$. The stationary distribution~$\pi$ satisfies~$\pi^T = \pi^T G$. 
A Markov chain in which~$\mu_t$ converges to~$\pi$ is called ``ergodic". We ensure~$P(\D, \rho)$ is ergodic via Assumptions~\ref{ass:rho} and~\ref{ass:N}
below. First we establish a key definition. 

To analyze the transient behavior of a Markov chain, it is common to analyze the convergence rate to~$\pi$. Many metrics exist for doing so, and we consider the ``$\infty$-norm ergodicity coefficient" of ergodic Markov chains. See~\cite{ipsen2011ergodicity} for a detailed account of ergodicity coefficients.
\begin{definition}[$\infty$-norm Ergodicity Coefficient; \citep{ipsen2011ergodicity}]\label{def:ergo}
    Let~$G\in\mathbb{S}_n$ be the transition matrix of an ergodic Markov chain with stationary distribution~$\pi$. Then the~$\infty$-norm ergodicity coefficient of the Markov chain is defined as
        $\tau_{\infty}(G) = \max_{\substack{\norm{z}_{\infty}= 1, z^T\ones = 0}}\norm{G^Tz}_{\infty}.$
   \hfill$\blacktriangle$
\end{definition}
Let the eigenvalue moduli of~$G\in\mathbb{S}_n$ be~$1=|\lambda_1|\geq \ldots\geq |\lambda_n|$. The eigenvalue modulus~$|\lambda_2|$ governs the rate of convergence of~$G$ to its stationary distribution if~$G$ is ergodic. Definition~\ref{def:ergo} provides a bound on~$|\lambda_2|$ when it is difficult to compute exactly. 

To solve Problem~\ref{prob:MC_utility}, we compute the total variation distance in the stationary distribution between a private Markov chain~$\tilde{P}$ and its non-private 
form~$P(\D, \rho)$, as well as the absolute difference in their $\infty$-norm ergodicity coefficients. For a Markov chain to converge to its stationary distribution, we require it to be (i) finite, (ii) irreducible, and (iii) homogeneous. Condition (i) holds trivially since the state space has the same cardinality as~$\rho$, which is finite. Condition (iii) holds since the database is not a function of time. Let the state space of the Markov chain model be~$\mathcal{Y}$, and note that each category~$\rho_i\in\rho$ corresponds to a state~$y_i\in\mathcal{Y}$. We apply the following assumptions on~$\mathcal{Y}$ and~$\mathcal{D}$ to ensure condition (ii) holds:
\begin{assumption}\label{ass:rho}
    Each state~$y_i\in\mathcal{Y}$ is reachable from any other state~$y_j\in\mathcal{Y}$ in finite time. 
\end{assumption}
\begin{assumption}\label{ass:N}
    In each~$D_i$, every feasible transition has occurred at least once.
\end{assumption}

Returning to the traffic scenario in Example~\ref{ex:traffic}, Assumption~\ref{ass:rho} enforces that all roads whose data are in~$\D$ can be reached in finite time from any other road with data in~$\D$. Assumption~\ref{ass:N} enforces that every possible turn on the roads with data in~$\D$ has been taken at least once. 


The next two theorems both rely on a bound on~$\|P(\mathcal{D}, \rho)-\tilde{P}\|_1$, which results in them having some terms in common. 
See the Supplementary Material, Section~\ref{ap:tech_lemmas} for the derivation of a bound on~$\|P(\mathcal{D}, \rho)-\tilde{P}\|_1$.
Given a database and a category set that satisfy Assumptions~\ref{ass:rho} and~\ref{ass:N}, we first bound the change in the stationary distribution, thus solving the first part of Problem~\ref{prob:MC_utility}. 

\begin{theorem}[Solution to Problem~\ref{prob:MC_utility}, Part 1]\label{thm:stat}
Let~$\tilde{\pi}$ denote the stationary distribution of~$\tilde{P}$.
Let Assumptions~\ref{ass:sat}-\ref{ass:N} hold. Then~$P(D, \rho)$ converges to a stationary distribution~$\pi$.
Let~$Z = (I-P(\D, \rho)-\ones_n\pi^T)^{-1}$ and let~$\zeta_i(x) = \log\big(\frac{x+1}{N_i}\big)-\psi\big(\frac{(x+1)k_i}{N_i}\big)$ for all~$i\in[n]$. Then we have
        $\E{\norm{\pi-\tilde{\pi}}_{TV}} \leq \frac{1}{2}\norm{Z}_1\sqrt{2L}$,    
    where    
        $L = \sum_{i=1}^n \pi(i)\big(\frac{n-1}{N_i}\zeta_i(0) +\frac{N_i-n+1}{N_i}\zeta_i(N_i-n) +\psi(k_i)\big).$    
\end{theorem}

Theorem~\ref{thm:stat} provides a tool for users to understand the quantitative change in the stationary distribution as both the number of individuals in the database and the strength of privacy change. Since the digamma function~$\psi(x)$ grows as~$\mathcal{O}(\log(x))$, Theorem~\ref{thm:stat} shows that the error in the stationary distribution is~$\mathcal{O}(\log(k_i^{-1}))$, 
which implies that accuracy can be substantially improved with small changes in privacy. Similarly, we develop a utility guarantee on the transient behavior (and solve the remaining part of Problem~\ref{prob:MC_utility}) by bounding the change in the ergodicity coefficient.

\begin{theorem}[Solution to Problem~\ref{prob:MC_utility}, Part 2]\label{thm:ergodic}
     Let the conditions of Theorem~\ref{thm:stat} hold. Let~$\tau_{\infty}(P(\D, \rho))$ be the ergodicity coefficient of~$P(\D, \rho)$ in the sense of Definition~\ref{def:ergo} and let~$\tau_{\infty}(\tilde{P})$ be the ergodicity coefficient of~$\tilde{P}$. Then
         $\E{|\tau_{\infty}(P(\D, \rho))-\tau_{\infty}(\tilde{P})|} \leq \sqrt{2L}.$     
\end{theorem}

Theorem~\ref{thm:ergodic}, similar to Theorem~\ref{thm:stat}, also provides users with a tool to understand how privacy affects the privatized Markov chain transitions and also behaves as~$\mathcal{O}(\log(k_i^{-1}))$. Combined, Theorems~\ref{thm:stat} and~\ref{thm:ergodic} provide insight into how the rate of convergence is affected by privacy and into the final error in the stationary distribution. 

\section{NUMERICAL SIMULATIONS}\label{sec:sims}
In this section we apply Algorithm~\ref{algo:private_count} to a database of grades for an undergraduate course and Algorithm~\ref{algo:private_MC} to a database of taxi trips in New York City.

\subsection{Class Grade Distribution}\label{sims:class}
Many universities publish the distribution of grades for classes at the end of every semester, along with the instructor’s name, to help students assess the relative difficulty of a course with a particular instructor. However, given side information about a student's performance on individual assignments, observers such as classmates may be able to make inferences about that student's final grade from this published distribution, which indicates a need for privacy.

\begin{figure}
    \centering
    \input{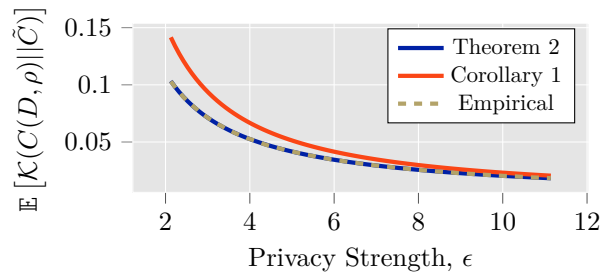}
    \caption{KL divergence between a privatized grade distribution and its true grade distribution. The upper bound in Corollary~\ref{cor:kld} becomes increasingly tight as privacy weakens, though it is close to the true value for all~$\epsilon$. The distributions of~$\tilde{\C}$ and~$\C(D, \rho)$ remain similar even under strong privacy, highlighting that Algorithm~\ref{algo:private_count} produces outputs with strong privacy that exhibit high accuracy.}
    \label{fig:dist}
\end{figure}

\begin{figure}
    \centering
    \includegraphics[width = \linewidth]{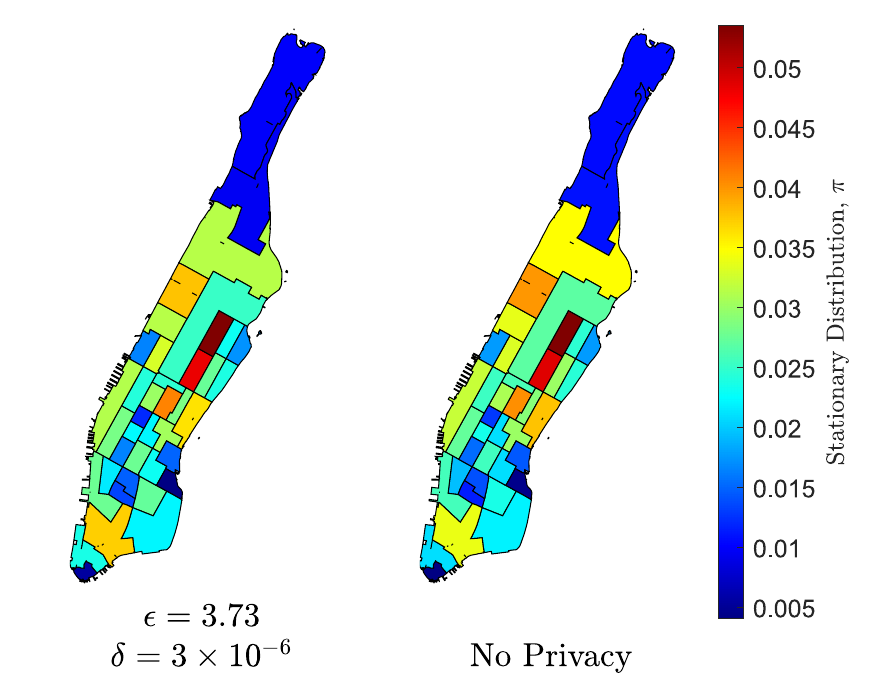}
    \caption{ Stationary distribution of the Markov chain formed by New York City taxi drop-offs and pickups. We see that, at the strongest privacy level~\new{$\epsilon = 3.73$}, the stationary distribution of the privatized Markov chain model remains close to that of the non-private model.}
    \label{fig:nyc_city}
\end{figure}

Using publicly available~\citep{gt2025grades} grade distributions~$\D$ and class sizes~$N$, we compute privatized forms of the grade distribution from a Calculus~$1$ course \new{with~$N=98$ students}. The categories are~$\rho = \{\texttt{A}, \texttt{B}, \texttt{C}, \texttt{D}, \texttt{F}\}$, and we take~$\eta = 0.073$ and~\new{$\gamma = 0.0004$}.  We simulate~$2,000$ privatized forms for each~$k \in \{20.6, 21.6,\ldots, 110.6\}$. The smallest value of~$k$ gives the strongest allowable privacy according to Assumption~\ref{ass:k}.  Figure~\ref{fig:dist} shows the average KL divergence between the private and non-private distributions with the strength of privacy~$\epsilon$ corresponding to each~$k$. We see that as privacy weakens (i.e., as~$\epsilon$ grows), the bound in Corollary~\ref{cor:kld} approaches the true average KL divergence.
Additionally, even at the \new{smallest~$\epsilon$ value, $\epsilon = 2.255$}, the private distributions remain accurate, with a KL divergence of~$0.103$. At this \new{smallest~$\epsilon$ value}, we find that~\new{$\delta = 0.0026$, and we can provide~$(2.255, 0.0026)$-differential privacy to this grade distribution. We may also achieve smaller~$\delta$ values at larger~$\epsilon$ values, e.g., at~$\epsilon = 3.31$, we obtain~$\delta = 1.3\times 10^{-4}$.} Thus, Algorithm~\ref{algo:private_count} produces private distributions 
with high similarity to the original, non-private distribution, which highlights the applicability to this scenario.

\subsection{New York City Taxis}\label{sims:nyc}

We apply Algorithm~\ref{algo:private_MC} to a Markov chain model of New York City Yellow Taxi trips. Each state in the Markov chain represents a pickup/drop-off zone, and the dataset in~\cite{nyc2025traffic} contains the pickup and drop-off locations of taxi trips in January 2025. We empirically compute the transition probabilities by counting the number of trips from a given pickup location to a corresponding drop-off location. Here,~$\rho_i$ is drop-off location~$i$, and~$D_j$ contains all departures from pickup location~$j$. We have~$d_k^j \models \rho_i$ if taxi trip~$k$ begins in pickup location~$j$ and drops off in drop-off location~$i$. See the Supplementary Material, Section~\ref{ap:sim_deets} for discussion on how the taxi data was pre-processed for this example.

To simulate~$1,000$ privatizations of the Markov chain, we find~$k_{min,i}$ via Assumption~\ref{ass:k} for every~$i\in[n]$, where~$n$ is the number of pickup/drop-off locations. Then we compute a set of~$k_i$ values by incrementing~$k_i$ by~$\frac{1}{10}k_{min,i}$ up to~$3k_{min,i}$ for all~$i\in[n]$. \new{Additionally, we set~$\gamma_i = 10^{-8}$ for all~$i\in[n]$.} We then find the corresponding strength of privacy~$\epsilon$ at each increment according to Algorithm~\ref{algo:private_MC} for the entire transition matrix. We find the strongest possible privacy to be~$(2.68, 0.0214)$-differential privacy. Figure~\ref{fig:nyc_city} illustrates the stationary distribution in the borough of Manhattan, \new{where there were~$N = 2,933,898$ taxi trips}, under different values of~$\epsilon$. As privacy weakens (i.e., as~$\epsilon$ grows), the stationary distribution approaches that of the original, non-private Markov chain. This is demonstrated numerically in Figure~\ref{fig:pi_diff}, which shows the expected~$TV$ distance between the stationary distributions with varying privacy strength. We find that even at the strongest possible privacy,\new{~$\epsilon = 3.73$ and~$\delta = 3\times 10^{-6}$}, the change induced by privacy on average is only~$\E{\norm{\pi-\tilde{\pi}}_{TV}} = 0.017$. We may attain even smaller~$\delta$ values by increasing the value of~$\epsilon$, e.g., we find~$\delta = 5\times 10^{-9}$ at~$\epsilon = 5.19$. Thus, we find that Algorithm~\ref{algo:private_MC} induces only small differences in the stationary distribution while provably protecting the data that is used to generate the underlying Markov chain model.

\begin{figure}
    \centering
%
%
\definecolor{chocolate2267451}{RGB}{226,74,51}
\definecolor{dimgray85}{RGB}{85,85,85}
\definecolor{gainsboro229}{RGB}{229,229,229}
\definecolor{lightgray204}{RGB}{204,204,204}
\definecolor{steelblue52138189}{RGB}{52,138,189}
\definecolor{black}{RGB}{0, 0, 0}
\definecolor{GTblue}{RGB}{0, 48, 87}
\definecolor{GTgold}{RGB}{179, 163, 105}
\definecolor{UFOrange}{RGB}{250, 70, 22}
\definecolor{UFblue}{RGB}{0, 33, 165}
\begin{tikzpicture}

\begin{axis}[%
width=0.35\figW,
height=0.4\figH,
axis background/.style={fill=gainsboro229},
axis line style={white},
scale only axis,
xlabel=\textcolor{black}{Privacy Strength, $\epsilon$},
xtick style={color=dimgray85},
x grid style={white},
yminorticks=true,
y grid style={white},
ylabel=\textcolor{black}{$\E{\norm{\pi-\tilde{\pi}}_{TV}}$},
xmajorgrids,
ymajorgrids,
yminorgrids,
tick align=outside,
tick pos=left,
legend pos = north west,
legend style={nodes={scale=0.85, transform shape}},
yticklabel style={
        /pgf/number format/fixed,
        /pgf/number format/precision=5
},
scaled y ticks=false,
]
\addplot [color=UFblue, ultra thick]
  table[row sep=crcr]{%
3.36589563406936	0.0183426257053171\\
3.69772856199766	0.0174932812748428\\
4.02961636524922	0.0166848895563111\\
4.36154735026649	0.016131474621404\\
4.69351289514246	0.0154397421262773\\
5.0255065100449	0.0150014929644126\\
5.35752322092319	0.0145799853073126\\
5.68955915422848	0.0140860056402454\\
6.02161125037876	0.0137158217376548\\
6.35367706190195	0.0133078484466752\\
6.68575460850613	0.0129864410716176\\
7.01784227129073	0.0126529023353829\\
7.34993871432718	0.0122291283362906\\
7.68204282571091	0.0121612185705429\\
8.01415367272599	0.0118943426715279\\
8.34627046729073	0.0115995156091728\\
8.67839253908843	0.011290215838717\\
9.01051931446349	0.0110709527621407\\
9.34265029964395	0.0109767130876933\\
9.67478506738378	0.0107781159994179\\
10.0069232461362	0.0105944471891985\\
};

\end{axis}
\end{tikzpicture}%
    \caption{Change in the stationary distribution with varying privacy strength. Even with
    \new{$\epsilon = 3.73$}, we find the average TV distance change in the stationary distribution is minimal. }
    \label{fig:pi_diff}
\end{figure}
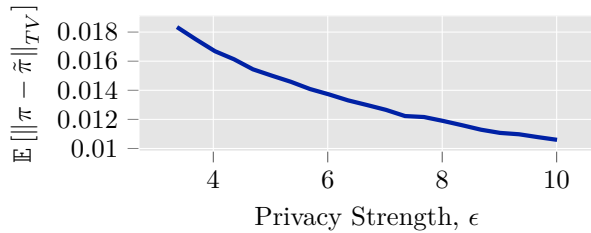

\section{Conclusion}
We have presented a framework privatizing databases that are used to compute both stochastic vectors and stochastic matrices. We proved that this approach is differentially private and provided utility guarantees that allow users to trade-off privacy, convergence rates, and accuracy of stationary distributions. \new{Results show with~$(3.73, 3\times 10^{-6})$-differential privacy, our approach induces less than~$2\%$ error in the stationary distribution on a large dataset $(N = 2,933,898)$}.  Future work \new{will address user-level privacy for data-driven Markov chain modeling and} privacy for data-driven Markov decision processes. 

\section*{Acknowledgements}
This work was partially supported by AFRL under grants FA8651-23-F-A008 and FA8651-24-1-0018, 
NSF under CAREER grant 1943275, ONR under grant N00014-21-1-2502, AFOSR under grant FA9550-19-1-0169, and NSF Graduate Research Fellowship under grant DGE-2039655. Any opinions, findings and conclusions or recommendations expressed in this material are those of the authors and do not necessarily reflect the views of sponsoring agencies.

\bibliography{main}

\begin{thebibliography}{38}
\providecommand{\natexlab}[1]{#1}
\providecommand{\url}[1]{\texttt{#1}}
\expandafter\ifx\csname urlstyle\endcsname\relax
  \providecommand{\doi}[1]{doi: #1}\else
  \providecommand{\doi}{doi: \begingroup \urlstyle{rm}\Url}\fi

\bibitem[Agarwal et~al.(2021)Agarwal, Kairouz, and Liu]{agarwal2021skellam}
Naman Agarwal, Peter Kairouz, and Ziyu Liu.
\newblock The skellam mechanism for differentially private federated learning.
\newblock \emph{Advances in Neural Information Processing Systems},
  34:\penalty0 5052--5064, 2021.

\bibitem[Benvenuti et~al.(2024{\natexlab{a}})Benvenuti, Bialy, Dennis, and
  Hale]{benvenuti2024guaranteed}
Alexander Benvenuti, Brendan Bialy, Miriam Dennis, and Matthew Hale.
\newblock Guaranteed feasibility in differentially private linearly constrained
  convex optimization.
\newblock \emph{IEEE Control Systems Letters}, 8:\penalty0 2745--2750,
  2024{\natexlab{a}}.

\bibitem[Benvenuti et~al.(2024{\natexlab{b}})Benvenuti, Hawkins, Fallin, Chen,
  Bialy, Dennis, and Hale]{benvenuti2023differentially}
Alexander Benvenuti, Calvin Hawkins, Brandon Fallin, Bo~Chen, Brendan Bialy,
  Miriam Dennis, and Matthew Hale.
\newblock Differentially private reward functions for markov decision
  processes.
\newblock In \emph{2024 IEEE Conference on Control Technology and Applications
  (CCTA)}, pages 631--636. IEEE, 2024{\natexlab{b}}.

\bibitem[Blanco-Justicia et~al.(2022)Blanco-Justicia, S{\'a}nchez,
  Domingo-Ferrer, and Muralidhar]{blanco2022critical}
Alberto Blanco-Justicia, David S{\'a}nchez, Josep Domingo-Ferrer, and
  Krishnamurty Muralidhar.
\newblock A critical review on the use (and misuse) of differential privacy in
  machine learning.
\newblock \emph{ACM Computing Surveys}, 55\penalty0 (8):\penalty0 1--16, 2022.

\bibitem[Chen et~al.(2023{\natexlab{a}})Chen, Hawkins, Karabag, Neary, Hale,
  and Topcu]{chen2023differential}
Bo~Chen, Calvin Hawkins, Mustafa~O Karabag, Cyrus Neary, Matthew Hale, and Ufuk
  Topcu.
\newblock Differential privacy in cooperative multiagent planning.
\newblock In \emph{Uncertainty in Artificial Intelligence}, pages 347--357.
  PMLR, 2023{\natexlab{a}}.

\bibitem[Chen et~al.(2023{\natexlab{b}})Chen, Leahy, Jones, and
  Hale]{chen2023differentialsymbolic}
Bo~Chen, Kevin Leahy, Austin Jones, and Matthew Hale.
\newblock Differential privacy for symbolic systems with application to markov
  chains.
\newblock \emph{Automatica}, 152:\penalty0 110908, 2023{\natexlab{b}}.

\bibitem[Chen et~al.(2022)Chen, Choo, Kairouz, and Suresh]{chen2022fundamental}
Wei-Ning Chen, Christopher A~Choquette Choo, Peter Kairouz, and Ananda~Theertha
  Suresh.
\newblock The fundamental price of secure aggregation in differentially private
  federated learning.
\newblock In \emph{International Conference on Machine Learning}, pages
  3056--3089. PMLR, 2022.

\bibitem[Chitra et~al.(2022)Chitra, Angeris, and Evans]{chitra2022differential}
Tarun Chitra, Guillermo Angeris, and Alex Evans.
\newblock Differential privacy in constant function market makers.
\newblock In \emph{International Conference on Financial Cryptography and Data
  Security}, pages 149--178. Springer, 2022.

\bibitem[Dong et~al.(2022)Dong, Li, Ma, and Liu]{dong2022personalized}
Jie Dong, Gui Li, Wenkai Ma, and Jianshun Liu.
\newblock Personalized recommendation system based on social tags in the era of
  internet of things.
\newblock \emph{Journal of Intelligent Systems}, 31\penalty0 (1):\penalty0
  681--689, 2022.

\bibitem[Dwork et~al.(2006)Dwork, McSherry, Nissim, and
  Smith]{dwork2006calibrating}
Cynthia Dwork, Frank McSherry, Kobbi Nissim, and Adam Smith.
\newblock Calibrating noise to sensitivity in private data analysis.
\newblock In \emph{Theory of cryptography conference}, pages 265--284.
  Springer, 2006.

\bibitem[Dwork et~al.(2014)Dwork, Roth, et~al.]{dwork2014algorithmic}
Cynthia Dwork, Aaron Roth, et~al.
\newblock The algorithmic foundations of differential privacy.
\newblock \emph{Foundations and Trends in Theoretical Computer Science},
  9\penalty0 (3--4):\penalty0 211--407, 2014.

\bibitem[Fallin et~al.(2023)Fallin, Hawkins, Chen, Gohari, Benvenuti, Topcu,
  and Hale]{fallin2023differential}
Brandon Fallin, Calvin Hawkins, Bo~Chen, Parham Gohari, Alexander Benvenuti,
  Ufuk Topcu, and Matthew Hale.
\newblock Differential privacy for stochastic matrices using the matrix
  dirichlet mechanism.
\newblock In \emph{2023 62nd IEEE Conference on Decision and Control (CDC)},
  pages 5067--5072. IEEE, 2023.

\bibitem[{Georgia Tech}(2025)]{gt2025grades}
{Georgia Tech}.
\newblock Course critique.
\newblock
  \url{https://lite.gatech.edu/lite_script/dashboards/grade_distribution.html},
  2025.

\bibitem[Geyer et~al.(2017)Geyer, Klein, and Nabi]{geyer2017differentially}
Robin~C Geyer, Tassilo Klein, and Moin Nabi.
\newblock Differentially private federated learning: A client level
  perspective.
\newblock \emph{arXiv preprint arXiv:1712.07557}, 2017.

\bibitem[Gohari et~al.(2020)Gohari, Hale, and Topcu]{gohari2020privacy}
Parham Gohari, Matthew Hale, and Ufuk Topcu.
\newblock Privacy-preserving policy synthesis in markov decision processes.
\newblock In \emph{2020 59th IEEE Conference on Decision and Control (CDC)},
  pages 6266--6271. IEEE, 2020.

\bibitem[Gohari et~al.(2021)Gohari, Wu, Hawkins, Hale, and
  Topcu]{gohari2021differential}
Parham Gohari, Bo~Wu, Calvin Hawkins, Matthew Hale, and Ufuk Topcu.
\newblock Differential privacy on the unit simplex via the dirichlet mechanism.
\newblock \emph{IEEE Transactions on Information Forensics and Security},
  16:\penalty0 2326--2340, 2021.

\bibitem[Gong et~al.(2011)Gong, Midlam-Mohler, Marano, and
  Rizzoni]{gong2011iterative}
Qiuming Gong, Shawn Midlam-Mohler, Vincenzo Marano, and Giorgio Rizzoni.
\newblock An iterative markov chain approach for generating vehicle driving
  cycles.
\newblock \emph{SAE International Journal of Engines}, 4\penalty0 (1):\penalty0
  1035--1045, 2011.

\bibitem[Hsu et~al.(2014)Hsu, Gaboardi, Haeberlen, Khanna, Narayan, Pierce, and
  Roth]{hsu2014differential}
Justin Hsu, Marco Gaboardi, Andreas Haeberlen, Sanjeev Khanna, Arjun Narayan,
  Benjamin~C Pierce, and Aaron Roth.
\newblock Differential privacy: An economic method for choosing epsilon.
\newblock In \emph{2014 IEEE 27th Computer Security Foundations Symposium},
  pages 398--410. IEEE, 2014.

\bibitem[Ipsen and Selee(2011)]{ipsen2011ergodicity}
Ilse~CF Ipsen and Teresa~M Selee.
\newblock Ergodicity coefficients defined by vector norms.
\newblock \emph{SIAM Journal on Matrix Analysis and Applications}, 32\penalty0
  (1):\penalty0 153--200, 2011.

\bibitem[Jayaraman and Evans(2019)]{jayaraman2019evaluating}
Bargav Jayaraman and David Evans.
\newblock Evaluating differentially private machine learning in practice.
\newblock In \emph{28th USENIX security symposium (USENIX security 19)}, pages
  1895--1912, 2019.

\bibitem[Kotz et~al.(2019)Kotz, Balakrishnan, and Johnson]{kotz2019continuous}
Samuel Kotz, Narayanaswamy Balakrishnan, and Norman~L Johnson.
\newblock \emph{Continuous multivariate distributions, Volume 1: Models and
  applications}, volume~1.
\newblock John wiley \& sons, 2019.

\bibitem[Lundstr{\"o}m et~al.(2016)Lundstr{\"o}m, J{\"a}rpe, and
  Verikas]{lundstrom2016detecting}
Jens Lundstr{\"o}m, Eric J{\"a}rpe, and Antanas Verikas.
\newblock Detecting and exploring deviating behaviour of smart home residents.
\newblock \emph{Expert Systems with Applications}, 55:\penalty0 429--440, 2016.

\bibitem[Machanavajjhala et~al.(2008)Machanavajjhala, Kifer, Abowd, Gehrke, and
  Vilhuber]{machanavajjhala2008privacy}
Ashwin Machanavajjhala, Daniel Kifer, John Abowd, Johannes Gehrke, and Lars
  Vilhuber.
\newblock Privacy: Theory meets practice on the map.
\newblock In \emph{2008 IEEE 24th international conference on data
  engineering}, pages 277--286. IEEE, 2008.

\bibitem[McSherry(2009)]{mcsherry2009privacy}
Frank~D McSherry.
\newblock Privacy integrated queries: an extensible platform for
  privacy-preserving data analysis.
\newblock In \emph{Proceedings of the 2009 ACM SIGMOD International Conference
  on Management of data}, pages 19--30, 2009.

\bibitem[Meiser(2018)]{meiser2018approximate}
Sebastian Meiser.
\newblock Approximate and probabilistic differential privacy definitions.
\newblock \emph{Cryptology ePrint Archive}, 2018.

\bibitem[Munoz et~al.(2021)Munoz, Syed, Vassilvtiskii, and
  Vitercik]{munoz2021private}
Andres Munoz, Umar Syed, Sergei Vassilvtiskii, and Ellen Vitercik.
\newblock Private optimization without constraint violations.
\newblock In \emph{International Conference on Artificial Intelligence and
  Statistics}, pages 2557--2565. PMLR, 2021.

\bibitem[Noble et~al.(2022)Noble, Bellet, and
  Dieuleveut]{noble2022differentially}
Maxence Noble, Aur{\'e}lien Bellet, and Aymeric Dieuleveut.
\newblock Differentially private federated learning on heterogeneous data.
\newblock In \emph{International Conference on Artificial Intelligence and
  Statistics}, pages 10110--10145. PMLR, 2022.

\bibitem[Norris(1998)]{norris1998markov}
James~R Norris.
\newblock \emph{Markov chains}.
\newblock Cambridge university press, 1998.

\bibitem[{NYC Taxi Limousine Commision}(2025)]{nyc2025traffic}
{NYC Taxi Limousine Commision}.
\newblock January 2025 yellow taxi trip record data.
\newblock \url{https://www.nyc.gov/site/tlc/about/tlc-trip-record-data.page},
  2025.

\bibitem[Ponnoprat(2023)]{ponnoprat2023dirichlet}
Donlapark Ponnoprat.
\newblock Dirichlet mechanism for differentially private kl divergence
  minimization.
\newblock \emph{Transactions on Machine Learning Research}, 2023.

\bibitem[Ponomareva et~al.(2023)]{ponomareva2023dp}
Natalia Ponomareva et~al.
\newblock How to dp-fy ml: A practical guide to machine learning with
  differential privacy.
\newblock \emph{Journal of Artificial Intelligence Research}, 77:\penalty0
  1113--1201, 2023.

\bibitem[Qi et~al.(2010)Qi, Guo, and Guo]{qi2010complete}
Feng Qi, Senlin Guo, and Bai-Ni Guo.
\newblock Complete monotonicity of some functions involving polygamma
  functions.
\newblock \emph{Journal of Computational and Applied Mathematics}, 233\penalty0
  (9):\penalty0 2149--2160, 2010.

\bibitem[Rendle et~al.(2010)Rendle, Freudenthaler, and
  Schmidt-Thieme]{rendle2010factorizing}
Steffen Rendle, Christoph Freudenthaler, and Lars Schmidt-Thieme.
\newblock Factorizing personalized markov chains for next-basket
  recommendation.
\newblock In \emph{Proceedings of the 19th international conference on World
  wide web}, pages 811--820, 2010.

\bibitem[Seneta(1988)]{seneta1988sensitivity}
E~Seneta.
\newblock Sensitivity to perturbation of the stationary distribution: some
  refinements.
\newblock \emph{Linear Algebra and its Applications}, 108:\penalty0 121--126,
  1988.

\bibitem[Wang and Choi(2023)]{wang2023information}
Youjia Wang and Michael~CH Choi.
\newblock Information divergences of markov chains and their applications.
\newblock \emph{arXiv preprint arXiv:2312.04863}, 2023.

\bibitem[Wang et~al.(2016)Wang, Hale, Egerstedt, and
  Dullerud]{wang2016differentially}
Yu~Wang, Matthew Hale, Magnus Egerstedt, and Geir~E Dullerud.
\newblock Differentially private objective functions in distributed cloud-based
  optimization.
\newblock In \emph{2016 IEEE 55th Conference on Decision and Control (CDC)},
  pages 3688--3694. IEEE, 2016.

\bibitem[Waterhouse(1983)]{waterhouse1983symmetric}
William~C Waterhouse.
\newblock Do symmetric problems have symmetric solutions?
\newblock \emph{The American Mathematical Monthly}, 90\penalty0 (6):\penalty0
  378--387, 1983.

\bibitem[Wid{\'e}n et~al.(2009)Wid{\'e}n, Nilsson, and
  W{\"a}ckelg{\aa}rd]{widen2009combined}
Joakim Wid{\'e}n, Annica~M Nilsson, and Ewa W{\"a}ckelg{\aa}rd.
\newblock A combined markov-chain and bottom-up approach to modelling of
  domestic lighting demand.
\newblock \emph{Energy and Buildings}, 41\penalty0 (10):\penalty0 1001--1012,
  2009.

\end{thebibliography}

\newpage
\onecolumn
\appendix 
\section{Computing~\texorpdfstring{$\delta$}{delta}}\label{ap:delta}
In this section, we detail our approach for computing~$\delta$ when~$\D = \{D\}$. To enforce~$(\epsilon, \delta)$-probabilistic differential privacy in the sense of Definition~\ref{def:pdp}, we partition the output space of the Dirichlet mechanism into two sets in via
\begin{equation}
    \Omega_1 = \{x\in\Delta_n\mid x_i\geq \gamma \text{ for all }i\in [n]\}
\end{equation}
and
\begin{equation}
    \Omega_2 = \{x\in\Delta_n \mid x\not\in \Omega_1\},
\end{equation}
as in Section~\ref{sec:dm-dp}.
We state the following lemma on the computation of~$\delta$.
\begin{lemma}[Computation of~$\delta$; \citep{gohari2021differential}]\label{lem:delta_fake}
    Fix a database~$D$, a category set~$\rho$ such that~$|\rho| = n\in\mathbb{N}$, a Dirichlet parameter~$k>0$, and a bordered simplex parameter~$\eta\in(0, 1)$. Let~$\C(D, \rho)\in\Delta_{n}^{(\eta)}$ be the count vector on the database~$D$ with category set~$\rho$. Let Assumptions~\ref{ass:sat} and~\ref{ass:gam} hold. Define
    \begin{equation}
        \mathcal{A}_n = \left\{x\in\mathbb{R}^{n-1}\mid \sum_{i \in [n - 1]}x_i\leq 1, x_i\geq \gamma \text{ for all } i\in [n-1] \right\}. 
    \end{equation}
    Then for a Dirichlet mechanism with parameter~$k > 0$, we have 
    \begin{equation}
        \prob{\dir{\C(D, \rho)}\in\Omega_1} = \frac{\int_{\mathcal{A}_{n}}\prod_{i\in [n-1]}x_i^{k\C(D, \rho)_i-1}\left(1-\sum_{i\in [n-1]}x_i\right)^{k(1-\sum_{i\in [n-1]}\C(D, \rho)_i)-1}\prod_{i\in [n-1]}dx_i}{B(k\C(D, \rho))}.
    \end{equation}
    \hfill$\blacklozenge$  
\end{lemma}
Since~$\C(D, \rho)$ appears in this computation of~$\delta$, this form of~$\delta$ cannot be used to guarantee differential privacy since the value of~$\delta$ must hold for all adjacent pairs of input data. However, we may minimize~$\prob{\dir{\C(D, \rho)}\in\Omega_1}$ over all~$\C(D, \rho)\in\Delta_{n}^{(\eta)}$ to find an upper bound on~$\delta$ which holds for all adjacent input data. That is, we find
\begin{align}
    &\prob{\dir{\C(D, \rho)}\in\Omega_2} = 1-\prob{\dir{\C(D, \rho)}\in\Omega_1}\leq 1-\min_{\C(D, \rho)\in\mathcal{C}(\rho)}\prob{\dir{\C(D, \rho)}\in\Omega_1} = \delta.\label{eq:compute_delta}
\end{align}
Next, we state a result showing that computation of~$\min_{\C(D, \rho)\in\mathcal{C}(\rho)}\prob{\dir{\C(D, \rho)}\in\Omega_1}$ can be performed efficiently.

\begin{lemma}[$\delta$ log-concavity; \citep{gohari2021differential}]\label{lem:delta_real}
    Let all of the conditions from Lemma~\ref{lem:delta_fake} hold. Then~$\prob{\dir{\C(D, \rho)}\in\Omega_1}$ is a log-concave function of~$\C(D, \rho)$ over the domain~$\Delta_{n}^{(\eta)}$.\hfill$\blacklozenge$  
\end{lemma}
Lemma~\ref{lem:delta_real} implies that the minimizer lies on the extreme points of the feasible region. As noted in~\citep{gohari2021differential}, Lemma~\ref{lem:delta_real} then implies that the time complexity of computing~$\delta$ is~$\mathcal{O}((n-1)^2)$. Thus, it remains tractable to compute~$\delta$ in this fashion. 
However, as~$n-1$ increases, the objective function in~\eqref{eq:delta} becomes increasingly computationally intensive to compute because it requires the evaluation of an increasing number of nested integrals. Many publicly available numerical integration libraries do not support calculations with an arbitrarily large number of nested integrals. For example, \texttt{MATLAB} only supports up to~$3$ integrals (though publicly shared frameworks exist supporting up to~$6$). As a result, we take a different approach to computing~$\delta$ when~$n-1$ is large. 

To address these issues, we introduce a Monte-Carlo approach to compute~$\delta$. Lemma~\ref{lem:delta_real} shows that we need only evaluate the objective function at the extreme points of the feasible region. Let~$\{q_1, \ldots,  q_{n-1}\}$ denote the set of extreme points of the feasible region. Using a Monte-Carlo approach, we approximate the probability in~\eqref{eq:delta} at each of the extreme points~$q_i$ by drawing samples of~$\dir{q_i}$ for all~$i\in [n-1]$ and evaluating if each sample is in the set~$\Omega_1$. The fraction of samples of~$\dir{q_i}$ in~$\Omega_1$ is then an approximation for~$\mathbb{P}(\dir{q_i}\in\Omega_1)$. We substitute the smallest approximated probability into~\eqref{eq:delta} to approximate~$\delta$. We denote this approximated value of~$\delta$ by~$\hat{\delta}$. The accuracy of this approach depends on the number of samples~$M$ used in the Monte Carlo simulation. Figure~\ref{fig:montecarlo} shows how the relative error between~$\delta$ and~$\hat{\delta}$ decreases with the number of samples. We see that with~$10^6$ samples,~$\hat{\delta}$ has less than~$1\%$ error, highlighting the practicality of this approach for estimating~$\delta$.

\begin{figure}
    \centering
%
%
\definecolor{chocolate2267451}{RGB}{226,74,51}
\definecolor{dimgray85}{RGB}{85,85,85}
\definecolor{gainsboro229}{RGB}{229,229,229}
\definecolor{lightgray204}{RGB}{204,204,204}
\definecolor{steelblue52138189}{RGB}{52,138,189}
\definecolor{black}{RGB}{0, 0, 0}
\definecolor{GTblue}{RGB}{0, 48, 87}
\definecolor{GTgold}{RGB}{179, 163, 105}
\definecolor{UFOrange}{RGB}{250, 70, 22}
\definecolor{UFblue}{RGB}{0, 33, 165}
\begin{tikzpicture}

\begin{axis}[%
width=0.35\figW,
height=0.56\figH,
axis background/.style={fill=gainsboro229},
axis line style={white},
scale only axis,
xlabel=\textcolor{black}{{Number of Samples, $M$}},
xtick style={color=dimgray85},
x grid style={white},
yminorticks=true,
y grid style={white},
ylabel=\textcolor{black}{$\frac{\E{|\delta -\hat{\delta}|}}{\delta}$},
xmajorgrids,
ymajorgrids,
yminorgrids,
tick align=outside,
tick pos=left,
legend pos = north east,
legend style={nodes={scale=0.85, transform shape}},
xmode = log,
ymode = log,
]

\addplot [color=UFOrange, ultra thick]
  table[row sep=crcr]{%
100	0.357260812935445\\
158	0.283597410268585\\
251	0.240733927848148\\
398	0.11620095693639\\
631	0.116048051407132\\
1000	0.10129874953897\\
1585	0.0840040480343751\\
2512	0.0610729883673316\\
3981	0.0566462304834216\\
6310	0.0445654400174083\\
10000	0.0341936125550702\\
15849	0.0303518539088299\\
25119	0.0246617776461312\\
39811	0.0210017580607945\\
63096	0.0159977594755668\\
100000	0.0120936308830086\\
158489	0.00916190475266702\\
251189	0.00748713219507049\\
398107	0.00574071877805488\\
630957	0.00480616654902993\\
1000000	0.00391951970765985\\
};
\end{axis}
\end{tikzpicture}%
    \caption{Average relative error between~$\delta$ and~$\hat{\delta}$.}
    \label{fig:montecarlo}
\end{figure}
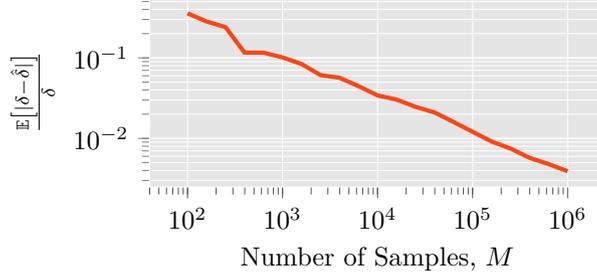

\section{Additional Technical Lemmas}\label{ap:tech_lemmas}
The omitted proofs in Section~\ref{ap:proofs} make use of the lemmas and definitions in this section. In addition to the special functions presented in Section~\ref{sec:intro}, we make use of the following in the subsequent definitions, lemmas, and proofs:
\begin{align}
        \psi^{(1)}(x) &= \psi'(x) 
    = \sum_{n=0}^{\infty} \frac{1}{(x+n)^2}\\
    \Beta(x; \alpha, \beta) &= \frac{x^{\alpha-1}(1-x)^{\beta-1}}{\Beta(\alpha, \beta)} \text{ with }\alpha>0,\beta>0,
\end{align}
where~$\psi^{(1)}$ denotes the trigamma function, and~$\Beta(x; \alpha_1, \beta)$ denotes the probability density function of the~$\Beta$ distribution.

First, we define the KL divergance for Markov chains.
\begin{definition}[Markov chain KL divergence]\label{def:MCKL}
    Let~$L, Q\in\mathbb{S}_n$ be two stochastic matrices with a common state space~$\mathcal{X}$. The Kullback-Leibler (KL) divergence from~$L$ to~$Q$  with respect to a probability distribution~$\pi$ on~$\mathcal{X}$ is
    \begin{equation}
        \kld(Q||L) = \sum_{i=1}^n \pi(i)\sum_{j=1}^n Q_{i j}\log\left(\frac{Q_{ij}}{L_{ij}}\right),
    \end{equation}
    where we use the convention~$0\log\big(\frac{0}{0}\big) = 0$. \hfill$\blacktriangle$
\end{definition}

\begin{lemma}[Markov chain divergence bound]\label{lem:mcklc}
    Fix Dirichlet parameters~$k_i>0$ for~$i\in[n]$. Fix a database~$\mathcal{D} = \{D_1, D_2,\ldots, D_n\}$ where each~$D_i$ has~$N_i\in\mathbb{N}$ entries, and category set~$\rho$ such that~$|\rho| = n\in\mathbb{N}$. Let~$P(\mathcal{D}, \rho)$ be defined as in~\eqref{eq:transitions_from_D}. Let Assumptions~\ref{ass:sat}-\ref{ass:N} hold so that~$P(\D, \rho)$ is a transition matrix of a finite, irreducible, aperiodic, time-homogeneous Markov chain. Let~$\tilde{P}$ denote the privatized form of~$P(\D, \rho)$ and let~$\pi$ denote the stationary distribution of~$P(\D, \rho)$. Let~$\kld(\cdot || \cdot)$ be defined as in Definition~\ref{def:MCKL}. Then
    \begin{equation}
        \E{\kld(P(D, \rho) || \tilde{P})} =  \sum_{i=1}^n \pi(i)\sum_{j=1}^n P(\D, \rho)_{i j}\big(\log(P(\D, \rho)_{ij})+\psi(k_i)-\psi(k_iP(\D, \rho)_{ij})\big).
    \end{equation}
    Additionally, we have that
    \begin{multline}\label{eq:mckld_bound}
        \E{\kld(P(D, \rho) || \tilde{P})} \\\leq \sum_{i=1}^n \pi(i)\left(\frac{n-1}{N_i}\left(\log\left(\frac{1}{N_i}\right)-\psi\left(\frac{k_i}{N_i}\right)\right) + \frac{N_i-n+1}{N_i}\left(\log\left(\frac{N_i-n+1}{N_i}\right) -\psi\left(\frac{(N_i-n+1)k_i}{N_i}\right)\right) +\psi(k_i)\right).
    \end{multline}
    \hfill$\blacklozenge$  
\end{lemma}

\emph{Proof.}
    From Definition~\ref{def:MCKL}, we have 
    \begin{equation}
    \kld(P(\D, \rho)||\tilde{P}) = \sum_{i=1}^n \pi(i)\sum_{j=1}^n P(\D, \rho)_{i j}\log\left(\frac{P(\D, \rho)_{ij}}{\tilde{P}_{ij}}\right).
    \end{equation}
Taking the expectation and exploiting linearity yields
\begin{equation}\label{eq:same_as_before}
    \E{\kld(P(\D, \rho)||\tilde{P})} = \sum_{i=1}^n \pi(i)\E{\sum_{j=1}^n P(\D, \rho)_{i j}\log\left(\frac{P(\D, \rho)_{ij}}{\tilde{P}_{ij}}\right)}.
\end{equation}
The expression
\begin{equation}
    \E{\sum_{j=1}^n P(\D, \rho)_{i j}\log\left(\frac{P(\D, \rho)_{ij}}{\tilde{P}_{ij}}\right)}
\end{equation}
is the expected KL divergence from~$P(\D, \rho)_i$, the~$i^{\text{th}}$ row of~$P(\D, \rho)$, to~$\tilde{P}_{i}$, i.e.,~$\E{\kld(P(\D, \rho)_i||\tilde{P}_i)}$. Substituting~$\E{\kld(P(\D, \rho)_i||\tilde{P}_i)}$ into~\eqref{eq:same_as_before} gives
\begin{equation}
    \E{\kld(P(\D, \rho)||\tilde{P})} = \sum_{i=1}^n \pi(i)\E{\kld(P(\D, \rho)_i||\tilde{P}_i)}.
\end{equation}
We note that~$\tilde{P}_{i}$ is a Dirichlet-distributed random vector, and thus substituting~\eqref{eq:kld} from Theorem~\ref{thm:kld} gives
\begin{equation}
    \E{\kld(P(\D, \rho)||\tilde{P})} = \sum_{i=1}^n \pi(i)\sum_{j=1}^n P(\D, \rho)_{i j}\big(\log(P(\D, \rho)_{ij})+\psi(k_i)-\psi(k_iP(\D, \rho)_{ij})\big), 
\end{equation}
as desired. 

Alternatively, for a bound that does not depend on the sensitive data~$\D$, we apply the upper bound on~$\C(D, \rho)_i\E{\log\left(\frac{\C(\D, \rho)_{i}}{\tilde{\C}_{i}}\right)}$ from~\eqref{eq:fstar} in the proof of Corollary~\ref{cor:kld} to obtain
\begin{multline}
   \E{\kld(P(\D, \rho)_i||\tilde{P}_i)} \\\leq \frac{n-1}{N_i}\left(\log\left(\frac{1}{N_i}\right)-\psi\left(\frac{k_i}{N_i}\right)\right) + \frac{N_i-n+1}{N_i}\left(\log\left(\frac{N_i-n+1}{N_i}\right) -\psi\left(\frac{(N_i-n+1)k_i}{N_i}\right)\right) +\psi(k_i),
\end{multline}
which we substitute into~\eqref{eq:same_as_before} to obtain
\begin{multline}
    \E{\kld(P(\D, \rho)||\tilde{P})} \\\leq \sum_{i=1}^n \pi(i)\left(\frac{n-1}{N_i}\left(\log\left(\frac{1}{N_i}\right)-\psi\left(\frac{k_i}{N_i}\right)\right) + \frac{N_i-n+1}{N_i}\left(\log\left(\frac{N_i-n+1}{N_i}\right) -\psi\left(\frac{(N_i-n+1)k_i}{N_i}\right)\right) +\psi(k_i)\right),
\end{multline}
which completes the proof.  \hfill$\square$

\begin{lemma}[Stationary Distribution Perturbation Bound; \citep{seneta1988sensitivity}]\label{lem:pert}
Let~$P$ be the transition matrix of a finite, homogeneous, irreducible Markov chain with stationary distribution~$\pi$, and let~$\tilde{P}$ be a perturbed form of~$P$. Let~$\tilde{\pi}$ be the stationary distribution of~$\tilde{P}$. Then,
\begin{equation}
    \norm{\pi-\tilde{\pi}}_1 \leq \norm{Z}_1\norm{P-\tilde{P}}_1,
\end{equation}
where~$Z = (I-P-\ones_n\pi^T)^{-1}$. \hfill$\blacklozenge$  
\end{lemma}

\begin{lemma}[$1$-norm Accuracy Bound]\label{lem:1norm}
    Let the conditions of Lemma~\ref{lem:mcklc} hold. Then
    \begin{multline}\label{eq:pinsk}
        \E{\norm{P-\tilde{P}}_{1}} \leq \sqrt{2}\cdot \\\left(\sum_{i=1}^n \pi(i)\left(\frac{n-1}{N_i}\left(\log\left(\frac{1}{N_i}\right)-\psi\left(\frac{k_i}{N_i}\right)\right) + \frac{N_i-n+1}{N_i}\left(\log\left(\frac{N_i-n+1}{N_i}\right) -\psi\left(\frac{(N_i-n+1)k_i}{N_i}\right)\right) +\psi(k_i)\right)\right)^{\frac{1}{2}}.
    \end{multline}
    \hfill$\blacklozenge$  
\end{lemma}
\emph{Proof}.
    From Pinsker's inequality (which also holds for matrix-valued~$f$-divergences~\citep{wang2023information}) it follows that
\begin{equation}
    \norm{P-\tilde{P}}_1 \leq \sqrt{2}\sqrt{\kld(P(\D, \rho)||\tilde{P})}.
\end{equation}
Taking the expectation on both sides gives 
\begin{equation}
    \E{\norm{P-\tilde{P}}_1} \leq \sqrt{2}\E{\sqrt{\kld(P(\D, \rho)||\tilde{P})}}\leq \sqrt{2}\sqrt{\E{\kld(P(\D, \rho)||\tilde{P})}},
\end{equation}
where the second inequality follows from Jensen's inequality and the fact that the square root function is concave. 
Substituting~\eqref{eq:mckld_bound} from Lemma~\ref{lem:mcklc} into the above bound gives
\begin{multline}
    \E{\norm{P-\tilde{P}}_{1}} \leq \sqrt{2}\cdot \\\left(\sum_{i=1}^n \pi(i)\left(\frac{n-1}{N_i}\left(\log\left(\frac{1}{N_i}\right)-\psi\left(\frac{k_i}{N_i}\right)\right) + \frac{N_i-n+1}{N_i}\left(\log\left(\frac{N_i-n+1}{N_i}\right) -\psi\left(\frac{(N_i-n+1)k_i}{N_i}\right)\right) +\psi(k_i)\right)\right)^{\frac{1}{2}},
\end{multline}
as desired.\hfill$\square$

\section{Omitted Proofs}\label{ap:proofs}

\subsection{Proof of Theorem~\ref{thm:dp}}\label{ap:coreps}
We seek to bound
\begin{equation}\label{eq:leq_epsilon2}
        \log\left( \frac{\prob{ \dir{\C(D, \rho)}=S }}
        {\prob{\dir{\C(D', \rho)}=S}} \right).
\end{equation}
The adjacency relation in Definition~\ref{def:adj} implies that adjacent databases~$D$ and~$D'$ have the same number of entries, one of which differs between them. Thus, there exists~$d\in D$ and~$d'\in D'$ such that~$d\neq d'$ and~$D\setminus\{d\} = D'\setminus\{d'\}$. Under Assumption~\ref{ass:sat},~$d\models \rho_i$ for exactly one~$i$, and~$d'\models \rho_j$ for exactly one~$j$. If~$i=j$, then~$\C(D, \rho) = \C(D', \rho)$ and the term in~\eqref{eq:leq_epsilon2} is~$0$. Suppose~$i\neq j$. Then~$\C(D, \rho)\neq\C(D', \rho)$ and in particular they differ in two entries. Without loss of generality, consider~$i, j\in [n]$ as the indices of the entries in which~$\C(D, \rho)$ and~$\C(D', \rho)$ differ. From the definition of the Dirichlet mechanism in Definition~\ref{def:dmech}, we have
\begin{align}
        \log\left( \frac{\prob{ \dir{\C(D, \rho)}=S }}
        {\prob{\dir{\C(D', \rho)}=S}} \right) &= \log\left(\frac{B(k\C(D', \rho))}{B(k\C(D, \rho))}\prod_{i=1}^n \frac{x_i^{k\C(D, \rho)_i-1}}{x_i^{k\C(D', \rho)_i-1}}\right) \\
        &= \log\left(\frac{\gammaf{k\C(D', \rho)_i}\gammaf{k\C(D', \rho)_j}x_i^{k\C(D, \rho)_i-1} x_j^{k\C(D, \rho)_j-1}}{\gammaf{k\C(D, \rho)_i}\gammaf{k\C(D, \rho)_j}x_i^{k\C(D', \rho)_i-1} x_j^{k\C(D', \rho)_j-1}}\right)\\
        &= \log\left(\frac{\gammaf{k\C(D', \rho)_i}\gammaf{k\C(D', \rho)_j}}{\gammaf{k\C(D, \rho)_i}\gammaf{k\C(D, \rho)_j}}\left(\frac{x_i}{x_j}\right)^{k(\C(D, \rho)_i-\C(D', \rho)_i)} \right),
\end{align}
where the final equality follows from the fact that~$\C(D, \rho)_i + \C(D, \rho)_j = \C(D', \rho)_i + \C(D', \rho)_j$. To bound the term in~\eqref{eq:leq_epsilon2}, we formulate the optimization problem
\begin{equation}\label{eq:epsilon_opt2}
        \begin{aligned}
        \nu = 
        &\begin{aligned}
            \underset{\C(D, \rho), \C(D', \rho), x}{\operatorname{max}} &\quad \log\left(\frac{\gammaf{k\C(D', \rho)_i}\gammaf{k\C(D', \rho)_j}}{\gammaf{k\C(D, \rho)_i}\gammaf{k\C(D, \rho)_j}}\left(\frac{x_i}{x_j}\right)^{k(\C(D, \rho)_i-\C(D', \rho)_i)}\right)
        \end{aligned}
            \\
            &\begin{aligned}
                \operatorname{subject} \operatorname{to } \,\,\quad&
                |\C(D, \rho)_i-\C(D', \rho)_i|\leq \frac{1}{N}\\
                &\C(D, \rho)_i + \C(D, \rho)_j = \C(D', \rho)_i + \C(D', \rho)_j\\ 
                &\C(D, \rho)_{i},\C(D, \rho)_{j} \in[\eta, 1-2\eta]\\
                &\C(D', \rho)_{i},\C(D', \rho)_{j} \in[\eta, 1-2\eta]\\
                &x_{i}, x_{j}\in[\gamma, 1-(n-1)\gamma],
            \end{aligned}
        \end{aligned}   
    \end{equation} 
where the first two constraints follow from adjacency and the remaining constraints enforce that~$\C(D, \rho), \C(D', \rho)\in\Delta_{n}^{(\eta)}$ and~$x\in\Omega_1$, where 
\begin{equation}
    \Omega_1 = \{x\in\Delta_n\mid x_i\geq \gamma \text{ for all }i\in [n]\}.
\end{equation}
Assumptions~\ref{ass:gam} and \ref{ass:simp} ensure that all intervals in the constraints are non-empty. Let~$\Xi$ denote the feasible region of~\eqref{eq:epsilon_opt2}. From the sub-additivity of the maximum, we have that
\begin{multline}\label{eq:nu}
        \nu \leq 
            \max_{\C(D, \rho), \C(D', \rho)_i, x\in\Xi}\log\left(\left(\frac{x_i}{x_j}\right)^{k(\C(D, \rho)_i-\C(D', \rho)_i)}\right)+ \\\max_{\C(D, \rho), \C(D', \rho), x\in\Xi}\quad \log\left(\frac{\gammaf{k\C(D', \rho)_i}\gammaf{k\C(D', \rho)_j}}{\gammaf{k\C(D, \rho)_i}\gammaf{k\C(D, \rho)_j}}\right) = \nu_1+\nu_2.
    \end{multline} 
    Starting with~$\nu_1$, we have 
    \begin{align}
        \nu_1 &= \max_{\C(D, \rho), \C(D', \rho), x\in\Xi}\log\left(\left(\frac{x_i}{x_j}\right)^{k(\C(D, \rho)_i-\C(D', \rho)_i)}\right)\\
        &\leq \max_{\C(D, \rho), \C(D', \rho), x\in\Xi}|k(\C(D, \rho)_i-\C(D', \rho)_i)|\left|\log\left(\frac{x_i}{x_j}\right)\right|\\
        & = \frac{k}{N}\log\left(\frac{1-(n-1)\gamma}{\gamma}\right).\label{eq:nu1_final}
    \end{align}
    Next we bound~$\nu_2$ from~\eqref{eq:nu}. Let~$p = \C(D, \rho)_i+\C(D, \rho)_j = \C(D', \rho)_i+\C(D', \rho)_j$. Then we have
\begin{equation}\label{eq:nu2}
        \begin{aligned}
        \nu_2 = 
        &\begin{aligned}
            \underset{\C(D, \rho), \C(D', \rho)}{\operatorname{max}} &\quad \log\left(\frac{\gammaf{k\C(D', \rho)_i}\gammaf{k(p-\C(D', \rho)_i)}}{\gammaf{k\C(D, \rho)_i}\gammaf{k(p-\C(D, \rho)_i)}}\right)
        \end{aligned}
            \\
            &\begin{aligned}
                \operatorname{subject} \operatorname{to } \,\,\quad&
                |\C(D, \rho)_i-\C(D', \rho)_i|\leq \frac{1}{N}\\
                &p\in[2\eta, 1-\eta]\\
                &\C(D, \rho)_{i},\C(D, \rho)_{j} \in[\eta, 1-2\eta]\\
                &\C(D', \rho)_{i},\C(D', \rho)_{j} \in[\eta, 1-2\eta].
            \end{aligned}
        \end{aligned}   
    \end{equation}

    Assume without loss of generality that~$\C(D', \rho)_i \leq \C(D, \rho)_i$ (since adjacency is a symmetric relation, the following also holds when~$\C(D, \rho)_i \leq \C(D', \rho)_i$). Then from~\citep[Lemma 3]{gohari2021differential}, 
    we can upper bound the objective
    in~\eqref{eq:nu2} by replacing~$p$ with its upper bound.
    Then we have 
    \begin{equation}\label{eq:nu3}
        \begin{aligned}
        \nu_2 \leq
        &\begin{aligned}
            \underset{\C(D, \rho), \C(D', \rho)}{\operatorname{max}} &\quad \log\left(\frac{\Beta(k\C(D', \rho)_i, k(1-\eta-\C(D', \rho)_i))}{\Beta(k\C(D, \rho)_i, k(1-\eta-\C(D, \rho)_i))}\right)
        \end{aligned}
            \\
            &\begin{aligned}
                \operatorname{subject} \operatorname{to } \,\,\quad&
                |\C(D, \rho)_i-\C(D', \rho)_i|\leq \frac{1}{N}\\
                &\C(D, \rho)_{i},\C(D, \rho)_{j} \in[\eta, 1-2\eta]\\
                &\C(D', \rho)_{i},\C(D', \rho)_{j} \in[\eta, 1-2\eta].
            \end{aligned}
        \end{aligned}   
    \end{equation} 
    The authors in~\citep{gohari2021differential} showed that 
    the Karush-Kuhn-Tucker (KKT) conditions for optimality are not satisfied in the interior of the feasible region 
    for a problem of the form
    of~\eqref{eq:nu3}. That same argument applies here, 
    and therefore only the extreme values of~$\C(D, \rho)_i$ and~$\C(D', \rho)_i$ need to be considered to find the optimum, thus giving the set
    of possible solutions 
    \begin{equation}
        \left\{\left(\eta, \eta+\frac{1}{N}\right), \left(1-2\eta, 1-2\eta-\frac{1}{N}\right), \left(\eta+\frac{1}{N}, \eta\right), \left(1-2\eta+\frac{1}{N}, 1-2\eta\right)\right\}.
    \end{equation}
    Substituting each solution into~\eqref{eq:nu3}, we find that~$\left(1-2\eta, 1-2\eta-\frac{1}{N}\right)$ and~$\left(1-2\eta+\frac{1}{N}, 1-2\eta\right)$ both yield negative objective values, while~$\left(\eta, \eta+\frac{1}{N}\right)$ and~$ \left(\eta+\frac{1}{N}, \eta\right)$ yield identical positive objective values. We select the optimal pair~$\big(\C(D, \rho)_i^*, \C(D', \rho)_i^*\big) = \left(\eta+\frac{1}{N}, \eta\right)$ to obtain
    \begin{equation}\label{eq:nu2_final}
        \nu_2 \leq \log\left(\frac{\Beta(k\eta, k(1-2\eta))}{\Beta\left(k\left(\eta+\frac{1}{N}\right), k(1-2\eta-\frac{1}{N})\right)}\right).
    \end{equation}
    Substituting~\eqref{eq:nu1_final} and~\eqref{eq:nu2_final} into~\eqref{eq:nu} gives
    \begin{equation}
        \nu \leq \log\left(\frac{\Beta(k\eta, k(1-2\eta))}{\Beta\left(k\left(\eta+\frac{1}{N}\right), k(1-2\eta-\frac{1}{N})\right)}\right) + \frac{k}{N}\log\left(\frac{1-(n-1)\gamma}{\gamma}\right). 
    \end{equation}
    \hfill$\square$
\subsection{Proof of Theorem~\ref{thm:kld}}\label{ap:thmkld}
 From the definition of the KL divergence, we have
    \begin{equation}
        \kld(\C(D, \rho) || \tilde{\C}) = \sum_{i=1}^n \C(D, \rho)_i\log\left(\frac{\C(D, \rho)_i}{\tilde{\C}_i}\right).
    \end{equation}
    Taking the expectation and using linearity yields
    \begin{equation}\label{eq:sub_integral_here}
        \E{\kld(\C(D, \rho) || \tilde{\C})} = \sum_{i=1}^n \C(D, \rho)_i\E{\log\left(\frac{\C(D, \rho)_i}{\tilde{\C}_i}\right)} = \sum_{i=1}^n \C(D, \rho)_i\left(\log\left(\C(D, \rho)_i\right) - \E{\log(\tilde{\C}_i)}\right).
    \end{equation}
    The marginal distribution for a single element of a Dirichlet random vector is the Beta distribution~\citep{kotz2019continuous}. The expectation of the log of a~$\Beta(x; \alpha, \beta)$ distributed random variable~$X$ is 
    \begin{equation}
        \E{\log(X)} = \psi(\alpha)-\psi(\alpha+\beta).
    \end{equation}
    Here,~$\tilde{C}_i\sim\Beta(x; k\C(D, \rho)_i, k-k\C(D, \rho)_i)$. Then we have
    \begin{equation}\label{eq:c_expectation}
        \E{\log\left(\tilde{\C}_i\right)} = \psi\big(k\C(D, \rho)_i\big) - \psi\big(k - k\C(D, \rho)_i +k\C(D, \rho)_i\big) = \psi(k\C(D, \rho)_i) - \psi(k).
    \end{equation}
    Substituting~\eqref{eq:c_expectation} into~\eqref{eq:sub_integral_here} gives us
    \begin{equation}
        \E{\kld(\C(D, \rho) || \tilde{\C})} = \sum_{i=1}^n \C(D, \rho)_i\big(\log(\C(D, \rho)_i)+\psi(k)-\psi(k\C(D, \rho)_i)\big)
    \end{equation}
    as desired. \hfill$\square$
    
\subsection{Proof of Corollary~\ref{cor:kld}}\label{ap:corkld}
To determine an upper bound that is independent of the sensitive data, we consider the optimization problem
    \begin{align}\label{opt:pstar}
        f^* = \max_{p}\quad &\E{\kld(p||\tilde{\C})}\\ 
        \text{subject to } \quad& p_i \in \left[\frac{1}{N}, \frac{N-n+1}{N}\right]\text{ for all }i\in[n]\\
        & \sum_{i=1}^n p_i = 1,
    \end{align}
    where the first constraint enforces that~$p$ satisfies Assumption~\ref{ass:sat} (i.e., it can be the output of~$\C(\cdot, \rho)$ for some some~$\D\in\mathfrak{D}$) and the second constraint enforces that~$p$ is still a probability mass function. It then follows that
    \begin{equation}\label{eq:plug_fstar}
        \E{\kld(\C(D, \rho)||\tilde{\C})}\leq f^*
    \end{equation}
    for all~$D\in\mathfrak{D}$, and we are left to solve the problem in~\eqref{opt:pstar} to obtain an upper bound. Substituting~\eqref{eq:kld} into the cost function in~\eqref{opt:pstar} gives us
        \begin{align}\label{opt:pstar2}
        f^* = \max_{p}\quad &\sum_{i=1}^n p_i(\log(p_i)-\psi(kp_i)+\psi(k))\\ 
        \text{subject to } \quad& p_i \in \left[\frac{1}{N}, \frac{N-n+1}{N}\right]\text{ for all }i\in[n]\\
        & \sum_{i=1}^n p_i = 1.
    \end{align}
 Next, we split the sum in the objective function into two terms:
    \begin{align}\label{opt:removeterm}
        f^* = \max_{p,q_2}\quad &\sum_{i=1}^{n-1} \big[p_i(\log(p_i)-\psi(kp_i)+\psi(k))\big] + q_2(\log(q_2)-\psi(kq_2)+\psi(k))\\ 
        \text{subject to } \quad& p_i, q_2 \in \left[\frac{1}{N}, \frac{N-n+1}{N}\right]\text{ for all }i\in[n-1]\\
        & \sum_{i=1}^{n-1} p_i +q_2 = 1.
    \end{align}
    From the sub-additivity of the maximum, we then have 
        \begin{align}\label{opt:subadd}
        f^* \leq \max_{p}\quad &\sum_{i=1}^{n-1} p_i(\log(p_i)-\psi(kp_i)+\psi(k)) + \max_{q_2} q_2(\log(q_2)-\psi(kq_2)+\psi(k))\\ 
        \text{subject to } \quad& p_i, q_2 \in \left[\frac{1}{N}, \frac{N-n+1}{N}\right]\text{ for all }i\in[n-1]\\
        &q_2 = 1-\sum_{i=1}^{n-1} p_i.
    \end{align}
    For the maximization over~$p$ we have
    \begin{align}\label{opt:left}
        f_1^{*} = \max_{p}\quad &\sum_{i=1}^{n-1} p_i(\log(p_i)-\psi(kp_i)+\psi(k)) \\ 
        \text{subject to } \quad& p_i \in \left[\frac{1}{N}, \frac{N-n+1}{N}\right]\text{ for all }i\in[n-1].
    \end{align}
    Note that
    \begin{equation}
        \sum_{i=1}^{n-1} p_i(\log(p_i)-\psi(kp_i)+\psi(k))
    \end{equation} is a symmetric function of the~$p_i$'s in the sense that an arbitrary permutation of the indices of~$\{1, 2, \ldots, n-1\}$ does not change its value. From~\citep{waterhouse1983symmetric}, it follows that every~$p_i$ in~\eqref{opt:left} must be equal at the optimum. Then the optimization problem in~\eqref{opt:left} becomes
    \begin{align}\label{opt:symmetric}
        f_1^{*} = f_1(q_1^*) =  (n-1)\max_{q_1}\quad &q_1(\log(q_1)-\psi(kq_1)+\psi(k)) \\ 
        \text{subject to } \quad& q_1 \in \left[\frac{1}{N}, \frac{N-n+1}{N}\right],
    \end{align}
    and we are left to determine what the value of~$q_1$ must be. To do so, we will show that~$f_1$ is monotonically decreasing 
    by analyzing the derivative of~$f_1$:
\begin{equation}\label{eq:kl_derivative}
        f_1'(q_1) = \frac{d}{dq_1} (q_1\log(q_1)-q_1\psi(kq_1)+q_1\psi(k)) = \log(q_1) +1 - \psi(kq_1)- kq_1\psi^{(1)}(kq_1)+\psi(k),
    \end{equation}
    where~$\psi^{(1)}(\cdot)$ is the trigamma function, and we neglect the constant~$(n-1)$ since it is a positive constant and has no impact on the sign of~$f'(q_1)$.
    Now define~$x = kq_1$. Noting that~$\log(q_1) = \log(x) - \log(k)$, we substitute~$x$ into~\eqref{eq:kl_derivative} to obtain
    \begin{equation}
        \log(q_1) +1 - \psi(kq_1)- kq_1\psi^{(1)}(kq_1)+\psi(k) = (\log(x) - \psi(x)- x\psi^{(1)}(x)+1)- (\log(k) - \psi(k)).
    \end{equation}
    We define~$H(x) = \log(x) - \psi(x)- x\psi^{(1)}(x)+1$, which we substitute above to obtain
    \begin{equation}\label{eq:make positive}
        f_1'(q_1) = H(kq_1) - (\log(k) - \psi(k)).
    \end{equation}
    We will now show that~$H(kq_1)$ is monotonically increasing in~$q_1$. Since~$\psi^{(1)}(x)$ is strictly negative and~$x>0$, we may rewrite~$H(x)$ as
    \begin{equation}
        H(x) = \log(x) + |x\psi^{(1)}(x)| - \psi(x) +1.
    \end{equation}
        We lower bound~$x = kq_1$ using Assumption~\ref{ass:k}. We know that~$\eta\leq \min_{i\in[n]}q_i$, and thus~$\eta^{-1}\geq \frac{1}{\min_{i\in[n]}q_i}$. Since~$k\geq \frac{3}{2\eta}$ from Assumption~\ref{ass:k}, substituting the bound~$\eta^{-1}\geq \frac{1}{\min_{i\in[n]}q_i}$ yields that~$k \geq \frac{3}{2}\frac{1}{\min_{i\in[n]}q_i}$. Additionally, since~$q_1\geq \min_{i\in[n]} q_i$, we have that~$x = kq_1\geq \frac{3}{2}\frac{1}{\min_{i\in[n]}q_i}\min_{i\in[n]} q_i = \frac{3}{2}$, implying~$\psi(x)>0$. Thus, we may rewrite~$H(x)$ again as
    \begin{equation}
        H(x) = \log(x) + |x\psi^{(1)}(x)| - |\psi(x)| +1.
    \end{equation}
    From~\citep[Theorem 2]{qi2010complete}, the function~$-|x\psi^{(1)}(x)| + |\psi(x)|$ is monotonically decreasing. Then~$|x\psi^{(1)}(x)| - |\psi(x)|$ is monotonically increasing. Since~$\log(x)$ is also monotonically increasing and the sum of monotonically increasing functions is monotonically increasing, $H(x)$ is monotonically increasing.    
    Since~$q_1<1$, it follows that~$kq_1<k$, and since~$H(x)$ is monotonically increasing, $H(kq_1) < H(k)$. Substituting this bound into~\eqref{eq:make positive} yields
    \begin{equation}
        f_1'(q_1) < H(k) - (\log(k) - \psi(k)) = \log(k) - \psi(k)- k\psi^{(1)}(k)+1 - \log(k) +\psi(k) = 1- k\psi^{(1)}(k).
    \end{equation}
    We are left to show that~$1-k\psi^{(1)}(k)<0$ or equivalently~$k\psi^{(1)}(k)>1$.
    Recall the integral representation of the trigamma function:
    \begin{equation}\label{eq:psi_int}
        \psi^{(1)}(k) = \int_0^{\infty} \frac{te^{-xt}}{1-e^{-t}}dt.
    \end{equation}
    We may bound the integrand using the fact that~$1-e^{-t} < t$ for~$t \in (0, \infty)$, which implies that
    \begin{equation}\label{eq:integrand_ineq}
        \frac{1}{1-e^{-t}}>\frac{1}{t}.
    \end{equation}
    Substituting~\eqref{eq:integrand_ineq} into~\eqref{eq:psi_int} yields
    \begin{equation}
        k\psi^{(1)}(k) = k\int_0^{\infty} \frac{te^{-kt}}{1-e^{-t}}dt > \int_0^{\infty} e^{-kt}dt = k\left(\frac{1}{k}\right) = 1.
    \end{equation}
    Thus,~$f_1'(q_1)<0$ and~$f_1(q_1)$ is monotonically decreasing in~$q_1$. Therefore,~$q_1^*$ is equal to the lower bound of the domain of~$q_1$, which is~$\frac{1}{N}$. Additionally, since~$q_2 = 1-\sum_{i=1}^{n-1} p_i = 1-(n-1)q_1$, we have~$q_2^* = 1-\frac{n-1}{N}$.
    This gives~$p^* = [\frac{1}{N}{\ones_{n-1}}^T,~1-\frac{n-1}{N}]^T$. 
    Thus, 
    \begin{equation}\label{eq:fstar}
        f^* \leq \frac{n-1}{N}\left(\log\left(\frac{1}{N}\right)-\psi\left(\frac{k}{N}\right)\right) + \frac{N-n+1}{N}\left(\log\left(\frac{N-n+1}{N}\right) -\psi\left(\frac{(N-n+1)k}{N}\right)\right) +\psi(k).
    \end{equation}
    Substituting~\eqref{eq:fstar} into~\eqref{eq:plug_fstar} gives
\begin{equation}\label{eq:sub_zeta_here}
        \E{\kld(\C(D, \rho)||\tilde{\C})}\leq \frac{n-1}{N}\left(\log\left(\frac{1}{N}\right)-\psi\left(\frac{k}{N}\right)\right) + \frac{N-n+1}{N}\left(\log\left(\frac{N-n+1}{N}\right) -\psi\left(\frac{(N-n+1)k}{N}\right)\right) +\psi(k).
    \end{equation}
    Defining~$\zeta(x) = \log\big(\frac{x+1}{N}\big)-\psi\big(\frac{(x+1)k}{N}\big) $, substituting in~$\zeta(0)$ and~$\zeta(N-n)$ completes the proof.\hfill $\square$

\subsection{Proof of Lemma~\ref{lem:parallel-comp-conventional}}
    Since the database we consider is~$\D = \{D_1, \ldots, D_m\}$, the output space~$\mathcal{S}$ from Definition~\ref{def:pdp} may be partitioned into the output space of each mechanism~$\mathcal{M}_i(D_i)$ such that~$\mathcal{S}=\mathcal{S}_1\times\cdots\times \mathcal{S}_m$, where~$\mathcal{S}_i = \text{Range}(\mathcal{M}_i(D_i))$. From Definition~\ref{def:pdp}, 
    the mechanism~$\mathcal{M}(\D) = (\mathcal{M}_1(D_1), \ldots, \mathcal{M}_m(D_m))$ satisfies~$(\max_{i \in [m]}\epsilon_i, \max_{i \in [m]}\delta_i)$-probabilistic differential privacy if 
    the output space~$\mathcal{S}$ may be partitioned into two disjoint sets~$\Omega_1$ and~$\Omega_2$
    such that for all $\D$,
    \begin{equation}
        \prob{\mathcal{M}(\D)\in\Omega _2}\leq\max_{i \in [m]}\delta_i,\end{equation}
    %
    and for all~$\D'$ adjacent to~$\D$ and all~$S\in\Omega _1$,
    \begin{equation}
        \frac{\prob{ \mathcal{M}(\D)=S }}
        {\prob{\mathcal{M}(\D')=S}} \leq\max_{i \in [m]}e^{\epsilon_i}.
    \end{equation}
    For all~$i \in [m]$, since~$\mathcal{M}_i$ satisfies Definition~\ref{def:pdp}, the set~$\mathcal{S}_i$ may be partitioned into some~$\Omega_1^i$ and~$\Omega_2^i$ such that 
    \begin{equation}
        \prob{\mathcal{M}_i(D_i)\in\Omega _2^i}\leq\delta_i,
    \end{equation}    
    and for all~$\D'$ adjacent to~$\D$ and all~$S_i\in\Omega_1^i$,
    \begin{equation}
        \frac{\prob{ \mathcal{M}_i(D_i)=S_i }}
        {\prob{\mathcal{M}_i(D'_i)=S_i}} \leq e^{\epsilon_i}.
    \end{equation}
    Then we define~$\Omega_1 = \Omega^1_1 \times \cdots \times \Omega^m_1$ and~$\Omega_2 = \Omega^1_2 \times \cdots\times \Omega^m_2$.     

    Next we show that~$\frac{\prob{ \mathcal{M}(D)=S }}
        {\prob{\mathcal{M}(D')=S}}\leq\max_{i\in[m]}e^{\epsilon_i}$.  
        For all~$S\in\Omega_1$,  
        \begin{align}
            \frac{\prob{\mathcal{M}(D) = S}}{\prob{\mathcal{M}(D') = S}} = \prod_{i = 1}^m \frac{\prob{\mathcal{M}_i(D_i) = S_i}}{\prob{\mathcal{M}_i(D_i') = S_i}}, 
        \end{align}
        where~$S = (S_1, \ldots, S_m)$ and $S_i\in\Omega_1^i$.         
        Since~~$\adj(\D, \D')=1$,~$\D$ and~$\D'$ vary in only~$1$ entry. Since the~$D_i$ are disjoint, we assume that the set that contains the differing element is~$D_j$. Then~$\prob{\mathcal{M}_i(D_i')= S_i} = \prob{\mathcal{M}_i(D_i)= S_i}$ for all~$i \in [m]\setminus \{j\}$. The product above then becomes
        \begin{equation}
            \prod_{i = 1}^m\frac{\prob{\mathcal{M}_i(D_i) = S_i}}{\prob{\mathcal{M}_i(D_i') = S_i}} = \frac{\prob{\mathcal{M}_j(D_j) = S_j}}{\prob{\mathcal{M}_j(D_j') = S_j}}.
        \end{equation}
        Since~$\mathcal{M}_{j}(D_j)$ satisfies~$(\epsilon_j, \delta_j)$-probabilistic differential privacy, it follows that
        \begin{equation}\label{eq:must_hold}
            \frac{\prob{\mathcal{M}(D) = S}}{\prob{\mathcal{M}(D') = S}} = \frac{\prob{\mathcal{M}_{j}(D_j) = S_j}}{\prob{\mathcal{M}_j(D_j') = S_j}}\leq e^{\epsilon_j}.
        \end{equation}
        Since~\eqref{eq:must_hold} must hold for any~$j \in [m]$, the largest upper bound is~$\max_{i \in [m]}e^{\epsilon_i}$, as desired.
    
    Now we show that the condition~$\prob{\mathcal{M}(D)\in\Omega _2} \leq \delta$ holds. From the same argument above,
\begin{equation}\label{eq:conventional-delta}
        \mathbb{P}(\mathcal{M}(\D)\in \Omega_2)=\prod_{i\in[m]}\mathbb{P}(\mathcal{M}_i(D_i)\in \Omega_2^i) 
        \leq\prod_{i\in[m]} \delta_i \leq \max_{i \in [m]} \delta_i ,
    \end{equation}
    which holds because~$\delta_i \in [0, 1]$ for all~$i \in [m]$.
    Therefore, from Definition~\ref{def:pdp} and the fact that probabilistic differential privacy implies conventional differential privacy, it follows that the release of~$\mathcal{M}(\D) = (\mathcal{M}_1(D_1), \ldots, \mathcal{M}_m(D_m))$
    provides~$\D$ with conventional~$(\max_i \epsilon_{i \in [m]},\max_{i \in [m]}\delta_i)$-differential privacy. \hfill$\square$

    \subsection{Proof of Theorem~\ref{thm:mc_dp}}
From Theorem~\ref{thm:dp},~$\mathcal{M}_{\text{Dir}}^{(k_i)}(\C(D_i, \rho))$ is~$(\epsilon_i, \delta_i)$-differentially private for all~$i\in[n]$. Since the~$D_i$'s partition the database~$\D$,~$\tilde{P}$ is the output of the parallel composition of~$n$ mechanisms, where the~$i^{\text{th}}$ is~$(\epsilon_i, \delta_i)$-probabilistic differentially private. Then from Lemma~\ref{lem:parallel-comp-conventional} the computation of~$\tilde{P}$ is~$(\max_{i\in[n]}\epsilon_i, \max_{i\in[n]}\delta_i)$-differentially private. \hfill$\square$

\subsection{Proof of Theorem~\ref{thm:stat}}\label{ap:stat}
    From the definition of the total variation distance,
    \begin{equation}
        \E{\norm{\pi-\tilde{\pi}}_{TV}} = \frac{1}{2}\E{\norm{\pi-\tilde{\pi}}_{1}}.
    \end{equation}
    From Lemma~\ref{lem:pert},
    \begin{equation}\label{eq:sub_pinsk_here}
    \E{\norm{\pi-\tilde{\pi}}_1} \leq \norm{Z}_1\E{\norm{P-\tilde{P}}_1}.
\end{equation}
Substituting~\eqref{eq:pinsk} from Lemma~\ref{lem:1norm} into~\eqref{eq:sub_pinsk_here} gives
\begin{multline}
   \E{\norm{\pi-\tilde{\pi}}_{TV}} = \frac{1}{2}\E{\norm{\pi-\tilde{\pi}}_1} \leq \frac{\sqrt{2}}{2}\norm{Z}_1\cdot\\\left(\sum_{i=1}^n \pi(i)\left(\frac{n-1}{N_i}\left(\log\left(\frac{1}{N_i}\right)-\psi\left(\frac{k}{N_i}\right)\right) + \frac{N_i-n+1}{N_i}\left(\log\left(\frac{N_i-n+1}{N_i}\right) -\psi\left(\frac{(N_i-n+1)k_i}{N_i}\right)\right) +\psi(k_i)\right)\right)^{\frac{1}{2}}.
\end{multline}
Defining~$\zeta_i(x) = \log\big(\frac{x+1}{N_i}\big)-\psi\big(\frac{(x+1)k_i}{N_i}\big)$ and substituting in~$\zeta_i(0)$ and~$\zeta_i(N_i-n)$ above completes the proof.\hfill$\square$

\subsection{Proof of Theorem~\ref{thm:ergodic}}\label{ap:ergodic}
From~\citep[Theorem 4.1]{ipsen2011ergodicity}, for given stochastic matrices~$S$,~$S_1$, and~$S_2$, the~$\infty$-norm ergodicity coefficient has the following properties:
\begin{itemize}
    \item $\tau_{\infty}(S)\leq \norm{S}_1$\
    \item $|\tau_{\infty}(S_1)-\tau_{\infty}(S_2)|\leq \tau_{\infty}(S_1-S_2)$.
\end{itemize}
Combining these properties for~$P(\D, \rho)$ and~$\tilde{P}$,
\begin{equation}
    |\tau_{\infty}(P(\D, \rho))-\tau_{\infty}(\tilde{P})| \leq \norm{P(\D, \rho)-\tilde{P}}_1.
\end{equation}
Taking the expectation,
\begin{equation}\label{eq:bound_this_max}
     \E{|\tau_{\infty}(P(\D, \rho))-\tau_{\infty}(\tilde{P})|} \leq \E{\norm{P(\D, \rho)-\tilde{P}}_1}.
\end{equation}
Substituting~\eqref{eq:pinsk} from Lemma~\ref{lem:1norm} then gives
 \begin{multline}
        \E{|\tau_{\infty}(P(\D, \rho))-\tau_{\infty}(\tilde{P})|} \leq \sqrt{2}\cdot \\\left(\sum_{i=1}^n \pi(i)\left(\frac{n-1}{N_i}\left(\log\left(\frac{1}{N_i}\right)-\psi\left(\frac{k_i}{N_i}\right)\right) + \frac{N_i-n+1}{N_i}\left(\log\left(\frac{N_i-n+1}{N_i}\right) -\psi\left(\frac{(N_i-n+1)k_i}{N_i}\right)\right) +\psi(k_i)\right)\right)^{\frac{1}{2}}.
    \end{multline}
Substituting in~$\zeta_i(0)$ and~$\zeta_i(N_i-n)$ above completes the proof.\hfill$\square$


\section{Additional Accuracy Results}\label{app:brandon_results}
\begin{lemma}[Accuracy bounds; \citep{fallin2023differential}]\label{lem:brandon}
Fix a database~$D$ with~$N\in\mathbb{N}$ entries, a category set~$\rho$ such that~$|\rho| = n\in\mathbb{N}$, and~$\eta\in(0, 1)$. Let~$\C(D, \rho)\in \Delta_{n}^{(\eta)}$ be defined as in~\eqref{eq:count}, and let~$\tilde{C}$ be its privatized form. Let Assumptions~\ref{ass:sat}-\ref{ass:k} hold. Then
    \begin{equation}
    \E{|\C(D, \rho)_i - \tilde{\C}_i|} = 2\frac{\C(D, \rho)_i^{k\C(D, \rho)_i}(1-\C(D, \rho)_i)^{k(1-\C(D, \rho)_i)}}{k\Beta(k\C(D, \rho)_i, k(1-\C(D, \rho)_i))},
        \quad \E{|\C(D, \rho)_i - \tilde{\C}_i|}\leq \frac{\gammaf{k}2^{1-k}}{\gammaf{\frac{k}{2}}^2k}
    \end{equation}
    and
    \begin{equation}
        \E{|\C(D, \rho)_i - \tilde{\C}_i|^2} =\frac{\C(D, \rho)_i-\C(D, \rho)_i^2}{k+1} ,\quad \E{|\C(D, \rho)_i - \tilde{\C}_i|^2}\leq \frac{1}{4(k+1)}.
    \end{equation}
    \hfill$\blacklozenge$
\end{lemma}

\begin{lemma}[Accuracy Lower Bound]\label{lem:ex_lower}
    Let the conditions of Lemma~\ref{lem:brandon} hold. Then
    \begin{equation}
        \E{|\C(D, \rho)_i - \tilde{\C}_i|}\geq 2\frac{\left(\frac{1}{N}\right)^{k(1-\frac{n-1}{N})}\left(\frac{n-1}{N}\right)^{k(1-\frac{1}{N})}}{k\Beta\left(\frac{k}{N}, \frac{k(n-1)}{N}\right)}.
    \end{equation}
    \hfill$\blacklozenge$
\end{lemma}
\emph{Proof}.
    From Lemma~\ref{lem:brandon},
    \begin{equation}\label{eq:expectation_eq}
        \E{|\C(D, \rho)_i - \tilde{\C}_i|} = 2\frac{\C(D, \rho)_i^{k\C(D, \rho)_i}(1-\C(D, \rho)_i)^{k(1-\C(D, \rho)_i)}}{k\Beta(k\C(D, \rho)_i, k(1-\C(D, \rho)_i))}.
    \end{equation}
    To lower bound~\eqref{eq:expectation_eq}, we seek a lower bound on the numerator and an upper bound on the denominator. We begin with a lower bound on the numerator. For all~$i\in[n]$, 
    \begin{equation}
        \frac{1}{N} \leq \C(D, \rho)_i \leq 1-\frac{n-1}{N}
    \end{equation}
    and
    \begin{equation}
        1-\frac{1}{N}\geq 1-\C(D, \rho)_i\geq \frac{n-1}{N}.
    \end{equation}
    For the numerator, 
    \begin{equation}\label{eq:num}
        \C(D, \rho)_i^{k\C(D, \rho)_i}(1-\C(D, \rho)_i)^{k(1-\C(D, \rho)_i)}\geq \left(\frac{1}{N}\right)^{k(1-\frac{n-1}{N})}\left(\frac{n-1}{N}\right)^{k(1-\frac{1}{N})}.
    \end{equation}
    Next, we upper bound the denominator. We proceed by showing that the denominator is monotonically decreasing in~$\C(D, \rho)_i$. Consider the following partial derivatives of the~$\Beta$ function:
    \begin{align}
        \frac{\partial}{\partial a}\Beta(a, b) &= \Beta(a, b)(\psi(a)-\psi(a+b))\label{eq:partiala}\\
        \frac{\partial}{\partial b}\Beta(a, b) &= \Beta(a, b)(\psi(b)-\psi(a+b)).\label{eq:partialb}
    \end{align}
    For all~$a, b>0$, $\Beta(a, b)>0$. Additionally,~$\psi(x)$ is strictly increasing for~$x>0$. Then~$\psi(a)-\psi(a+b)<0$ and~$\psi(b)-\psi(a+b)<0$. These inequalities imply that~\eqref{eq:partiala} and~\eqref{eq:partialb} are negative for all~$a, b>0$ and~$\Beta(a, b)$ is therefore monotonically decreasing in both~$a$ and~$b$. We then upper bound the denominator of~\eqref{eq:expectation_eq} by
    \begin{equation}\label{eq:dom}
        k\Beta(k\C(D, \rho)_i, k(1-\C(D, \rho)_i)) \leq k\Beta\left(\frac{k}{N}, \frac{k(n-1)}{N}\right).
    \end{equation}
    We then substitute~\eqref{eq:num} and~\eqref{eq:dom} into~\eqref{eq:expectation_eq} to obtain
    \begin{equation}
        \E{|\C(D, \rho)_i - \tilde{\C}_i|} \geq  2\frac{\left(\frac{1}{N}\right)^{k(1-\frac{n-1}{N})}\left(\frac{n-1}{N}\right)^{k(1-\frac{1}{N})}}{k\Beta\left(\frac{k}{N}, \frac{k(n-1)}{N}\right)},
    \end{equation}
    which completes the proof.\hfill$\square$

\section{Additional Simulation Details}\label{ap:sim_deets}

In this section, we detail the pre-processing steps performed on the NYC taxi trips database in Section~\ref{sec:sims}. First, we considered only trips between taxi zones on the island of Manhattan since this is where the majority of taxi trips take place. Additionally, while some smaller islands are considered taxi zones in Manhattan, we neglect to include these since there is no clear interpretation of a taxi trip between islands where there is no bridge, such as from Manhattan to the Statue of Liberty. 
\begin{table}
\centering
\caption{The strongest possible privacy at~$\gamma = 10^{-8}$ for varying numbers of taxi zones, where there are originally~$63$ zones. As more zones are combined, users attain stronger privacy guarantees while still maintaining model accuracy.}
\begin{tabular}{|c|c|}
\hline
Number of Taxi Zones& Smallest Possible~$\epsilon$\\
\hline
$63$ & $45.39$\\
$60$ & $45.39$\\
$55$ & $19.82$\\
$50$ & $4.690$\\
$45$ & $4.690$\\
$40$ & $3.730$\\
$35$ & $3.078$\\
$30$ & $3.078$\\
\hline
\end{tabular}
\label{table:the_whole_table}
\end{table}

Some taxi zones are significantly larger than others, and due to the dependency of~$\epsilon$ on the number of participants in each independent subset the database~$N_i$, they achieve much stronger privacy than smaller taxi zones. As a result, smaller taxi zones have an out-sized impact on weakening the privacy guarantee of the entire Markov chain model. To that end, we combine low data taxi zones with adjacent taxi zones to achieve a more equitable distribution of participants. Table~\ref{table:the_whole_table} shows the strongest possible privacy~(i.e., smallest~$\epsilon$) for a varying number of taxi zones with~$\gamma = 10^{-8}$. As the number of zones decreases, it is possible to attain stronger privacy, highlighting a fundamental tradeoff: as a database is partitioned into smaller subsets with more categories, it becomes more difficult to guarantee strong privacy for the entire dataset. Thus, database curators interested in strong privacy must consider this tradeoff when selecting categories and partitioning a database.

\end{document}